\documentclass[a4paper,11pt]{article}
\usepackage{a4wide}
\usepackage{mathrsfs} 
\usepackage{amsmath}
\usepackage{amssymb}
\usepackage{amsfonts}
\usepackage{enumitem}
\usepackage{amsthm}
\usepackage{xcolor}
\usepackage{todonotes}
\usepackage{bm}
\usepackage{verbatim}
\usepackage{mathtools}
\usepackage{hyperref}
\hypersetup{colorlinks=true,citecolor=blue,linkcolor=blue,urlcolor=blue}
\usepackage{cleveref}
\usepackage{caption}
\usepackage{subcaption}
\usepackage{bbm}
\usepackage{stmaryrd}
\usepackage{thm-restate}

\newtheorem{theorem}{Theorem}[section]
\newtheorem{lemma}[theorem]{Lemma}
\newtheorem{corollary}[theorem]{Corollary}

\newtheorem*{claim*}{Claim}
\newtheorem*{theorem*}{Theorem}

\theoremstyle{definition}
\newtheorem{definition}[theorem]{Definition}

\newtheorem{remark}[theorem]{Remark}
\newtheorem*{remark*}{Remark}

\newcommand{\mnn}[1]{\mathscr{#1}}
\newcommand{\vmnn}[1]{\mathbb{#1}}
\newcommand{\tuple}[1]{{\mathbf{#1}}}

\newcommand{\proj}[2]{\mathrm{proj}^{#1}_{#2}}

\newcommand{\cc}{\mathbb{C}}

\newcommand{\ex}[1]{\mathbb{E}_{#1}}
\newcommand{\gr}{\mathscr{G}}
\newcommand{\sgr}{\mathscr{H}}

\newcommand{\bx}{\mathbf{x}}

\newcommand{\bfa}{\mathbf{A}}
\newcommand{\bfb}{\mathbf{B}}

\newcommand{\ba}{\mathbf{a}}
\newcommand{\bb}{\mathbf{b}}
\newcommand{\bc}{\mathbf{c}}
\newcommand{\bd}{\mathbf{d}}
\newcommand{\be}{\mathbf{e}}
\newcommand{\bg}{\mathbf{g}}
\newcommand{\bh}{\mathbf{h}}

\newcommand{\glc}{\mathrm{GLC}}
\newcommand{\abcs}{(\bfa,\bfb,c,s)}
\newcommand{\pcsp}{\mathrm{PCSP}}
\newcommand{\sig}{\sigma}
\newcommand{\im}{\textnormal{Im}}
\newcommand{\dom}{\textnormal{Dom}}

\newcommand{\Ind}{\mathrm{Ind}}
\newcommand{\GL}{\mathrm{GL}}
\newcommand{\groupid}{1}
\newcommand{\tsp}[1]{{#1}^{t}}
\newcommand{\tr}{\mathrm{tr}}

\newcommand{\triv}{1} 
\newcommand{\A}{\mathcal{A}}
\newcommand{\B}{\mathcal{B}}
\newcommand{\ar}{\mathrm{ar}}
\newcommand{\eq}{\ensuremath{\operatorname{3-LIN}}}

\newcommand{\eqqq}{\ensuremath{\operatorname{2-LIN}}}
\newcommand{\plu}{\textnormal{Plu}}

\newcommand{\FinSet}{\mathrm{FinSet}}
\newcommand{\OmegaI}{\Omega^{\mathrm{in}}}
\newcommand{\OmegaO}{\Omega^{\mathrm{out}}}
\newcommand{\IO}{\kappa}

\begin{document}

\title{Optimal Inapproximability of Promise Equations\\ over Finite Groups\thanks{An extended abstract of this work appeared in Proceedings of ICALP\,(A) 2025~\cite{BLZ25:icalp}. This work was supported by UKRI EP/X024431/1 and done when SB and AL were at the University of Oxford. For the purpose of Open Access, the authors have applied a CC BY public copyright licence to any Author Accepted Manuscript version arising from this submission. All data is provided in full in the results section of this paper.}}

\author{Silvia Butti\\
King's College London\\
\texttt{silvia.butti@kcl.ac.uk}
\and
Alberto Larrauri\\
University of Zaragoza\\
\texttt{llarrauri@unizar.es}
\and
Stanislav \v{Z}ivn\'y\\
University of Oxford\\
\texttt{standa.zivny@cs.ox.ac.uk}
}
\date{\today}

\maketitle

\begin{abstract}

A celebrated result of H{\aa}stad established that, for any constant $\varepsilon>0$,
it is NP-hard to find an assignment satisfying a $(1/|\gr|+\varepsilon)$-fraction
of the constraints of a given $\eq$ instance over an Abelian group $\gr$ even if
one is promised that an assignment satisfying a $(1-\varepsilon)$-fraction of
the constraints exists. Engebretsen, Holmerin, and Russell showed the same
result for $\eq$ instances over any finite (not necessarily Abelian) group.
In other words, for almost-satisfiable instances of $\eq$ the random
assignment achieves an optimal approximation guarantee.
We prove that the random assignment algorithm is still best possible under a stronger
promise that the $\eq$ instance is almost satisfiable over an arbitrarily
more restrictive group.

\end{abstract}

\section{Introduction}

The PCP
theorem~\cite{Arora98:jacm-probabilistic,Arora98:jacm-proof,Dinur07:jacm} is one
of the jewels of computational complexity and theoretical computer science more
broadly~\cite{Arora2009computational}. One of its equivalent statements is as
follows: The maximum number of simultaneously satisfiable constraints of a
Constraint Satisfaction Problem, or CSP for short, is NP-hard to approximate within some
constant factor. That is, while NP-hardness of CSPs means that it is NP-hard to
distinguish instances that are satisfiable from those that are unsatisfiable,
the PCP theorem shows that there is an absolute constant $\alpha<1$ such that it
is NP-hard to distinguish satisfiable CSP instances from those in which strictly
fewer than an $\alpha$-fraction of the constraints can be simultaneously
satisfied. Thus it is NP-hard to find an assignment that satisfies an
$\alpha$-fraction of the constraints even if one is promised that a satisfying
assignment exists. For some CSPs, as we shall see shortly, the optimal value of
$\alpha$ is known.

A classic example of a CSP is 3-SAT, the satisfiability problem of
CNF-formulas in which each clause contains 3 literals. The random assignment gives
a method to find an assignment that satisfies a $7/8$-fraction of the clauses.
H{\aa}stad famously showed that this is optimal in the following sense: For any
constant $\varepsilon>0$, it is NP-hard to find an assignment satisfying a
$(7/8+\varepsilon)$-fraction of the clauses of a 3-SAT instance even if one
is promised that a satisfying assignment exists~\cite{Hastad01:jacm}.

Another classic CSP is $\eq$,
the problem of solving linear equations in 3
variables over the Boolean domain $\{0,1\}$. If all equations can be satisfied
simultaneously then a satisfying assignment can be found in polynomial time by
Gaussian elimination. What can be done if no satisfying assignment exists? As
for 3-SAT, the random assignment gives a method to find a somewhat satisfying
assignment, namely one that satisfies a $1/2$-fraction of the constraints. As it
turns out, this is best possible even for instances of $\eq$ that are almost
satisfiable. In detail, H{\aa}stad showed that for any constant $\varepsilon>0$,
it is NP-hard to find an assignment satisfying a $(1/2+\varepsilon)$-fraction of
the constraints of a $\eq$ instance even if one is promised that an
assignment satisfying a $(1-\varepsilon)$-fraction of the constraints exists.
In fact, H{\aa}stad established optimal inapproximability results for $\eq$ over any finite Abelian group, not
just $\{0,1\}$. This result was later extended by Engebretsen, Holmerin, and
Russell to all finite groups~\cite{EHR04:tcs}. Since these foundational works,
Guruswami and Raghavendra~\cite{Guruswami09:toct} showed NP-hardness of finding
a barely satisfying assignment for a $\eq$ instance over the reals (and thus
also over the integers) even if a nearly satisfying assignment is promised to
exist over the integers. The same result was later established for $\eqqq$ for
large enough cyclic groups~\cite{ODonnell15:toct}. Khot and
Moshkovitz~\cite{Khot13:sicomp} studied inapproximability of $\eq$ over the
reals.

\medskip

In this work, we strengthen the optimal inapproximability results for $\eq$ over
finite groups by establishing NP-hardness of beating the random assignment threshold even if the instance is almost
satisfiable in an arbitrarily more restrictive setting. Formally, this is
captured by fixing (not one but) two groups and a homomorphism between
them, following the framework of promise CSPs~\cite{AGH17,BBKO21}. 
In detail, (decision) promise CSPs~\cite{BBKO21} can be seen as a qualitative
form of approximation: Each constraint comes in two forms, a strong one and a
weak one. The promise is that there is a solution satisfying all constraints in
the strong form while the (potentially easier) goal is to find a solution
satisfying all constraints in the weak form. An example of a strong vs.\ weak
constraint on the same, say Boolean, domain is $1$-in-$3$ vs NAE, where the former is $\{(0,0,1),(0,1,0),(1,0,0)\}$ and the latter is $\{(0,0,1),(0,1,0),(1,0,0),(1,1,0),(1,0,1),(0,1,1)\}$. NAE is weaker as the relation contains more tuples. 
While these two constraint relations capture the well-known NP-hard problems of
1-in-3-SAT and Not-All-Equal-SAT respectively~\cite{Schaefer78:stoc}, finding an
NAE-assignment turns out to be doable in polynomial time under the promise that
a 1-in-3-assignment exists~\cite{BG21}!\footnote{The algorithm involves solving
an instance of $\eq$ over the integers and rounding positive integers to $1$ and
non-positive integers to $0$, demonstrating the importance of $\eq$ among
promise CSPs.}
For constraints on different domains, the notion of strong vs.\ weak constraint is captured by a homomorphism between the (sets of all) constraint relations; in the example above, the homomorphism is just the identity function.
The exact solvability of $\eq$ in the promise setting was resolved in~\cite{LZ25:tocl}.

Recent work of Barto et al.~\cite{Barto24:lics} initiated an algebraic
investigation of (quantitative)
approximation of promise CSPs. In the context of $\eq$, 
here are two simple examples captured by this framework. First, let $\gr$ be a group and $\sgr$ be a subgroup of $\gr$. Given an almost-satisfiable system over the subgroup $\sgr$, maximise the number of satisfied equations over $\gr$. Our results imply that beating the random assignment over $\sgr$ is NP-hard. In the second example, consider a group $\gr$, a normal subgroup $\sgr$, and an almost-satisfiable system over $\gr$. The goal this time is to maximise the number of satisfied equations in the system over the quotient $\gr/\sgr$. Our results show that doing better than the random assignment over $\gr/\sgr$ is NP-hard. More generally, going beyond subgroups and quotients of a given group, we fix two groups $\gr_1$ and $\gr_2$ and a group homomorphism $\varphi$ from a subgroup $\sgr_1$ of $\gr_1$ to a subgroup $\sgr_2$ of $\gr_2$ with the property that $\varphi$ extends to a group homomorphism from $\gr_1$ to $\gr_2$. Given an almost-satisfiable system of equations over $\gr_1$ with constants in $\sgr_1$, the goal is to maximise the number of satisfied equations over $\gr_2$ where the constants are 
interpreted in $\sgr_2$ via $\varphi$. Our main result establishes that doing
better than the random assignment over $\sgr_2$ is NP-hard.

\begin{theorem*}[Main theorem, informal statement]
Let $\gr_1$ and $\gr_2$ be finite groups, $\sgr_1$ a subgroup of $\gr_1$, and $\varphi$ a group homomorphism from $\sgr_1$ to $\sgr_2=\im(\varphi)$, which is a subgroup of $\gr_2$.
For arbitrary $\epsilon,\delta>0$, given a system of linear equations over $\gr_1$ with constants in $\sgr_1$ that admits a solution satisfying at least a $(1-\epsilon)$-fraction of the constraints, it is NP-hard to find a solution over $\gr_2$ with constants in $\sgr_2$ that satisfies at least a  $(1/|\sgr_2|+\delta)$-fraction of the constraints, where the constants are interpreted via $\varphi$.
\end{theorem*}

For a precise statement of the main result, cf.~\Cref{th:main}. Thus we give an optimal inapproximability result for a natural and fundamental fragment of promise CSPs, systems of linear equations.

Other fragments of promise CSPs whose quantitative approximation has been studied includes almost approximate graph colouring~\cite{Engebretsen08:rsa,Dinur10:focs,Khot12:focs,Hecht23:approx}, approximate colouring~\cite{nz23:arxiv}, and approximate graph homomorphism~\cite{nz25:icalp}. Recent work of Brakensiek, Guruswami, and Sandeep~\cite{BGS23:stoc} studied robust approximation of promise CSPs; in particular, they observed that Raghavendra's celebrated theorem on approximate CSPs~\cite{Raghavendra08:everycsp} applies to promise CSPs, which in combination with the work of Brown-Cohen and Raghavendra~\cite{BCR15} gives an alternative framework for studying quantitative approximation of promise CSPs.

\medskip

The general approach for establishing inapproximability of systems of equations,
going back to~\cite{Hastad01:jacm,EHR04:tcs}, can be seen as a reduction from another CSP that is hard to approximate. In this reduction, one initially
transforms an instance of the original CSP to a system of equations of the
form $xyz=1$. To guarantee the soundness of this reduction, one needs to show that any
assignment that beats the random assignment in the target system of
equations can be transformed into a ``good'' assignment of the original
instance. To do this it is necessary to rule out vacuous assignments (e.g., the
assignment that sends all variables to the group identity) through a procedure
called folding, which introduces 
constants in the system of equations.
Afterwards, the soundness bounds are shown by performing Fourier analysis on
certain functions derived from the system.
Our proof follows this general approach. The main obstacle to
applying the techniques of~\cite{EHR04:tcs} directly is the fact that in our setting 
the constants lie in a proper subgroup of the
ambient group, which precludes us from applying classical folding over groups. Instead, we use a weaker notion of folding. This, however, implies that in the soundness analysis we have to take care of functions whose Fourier
expansion has non-zero value for the trivial term. To tackle this
issue, we consider the behaviour of irreducible group representations when they are
restricted to the subgroup of constants
via Frobenius Reciprocity.
We note that, as in the non-promise setting, the proof in the Abelian case is
much simpler, particularly using that the set of characters of an Abelian group
corresponds to its Pontryagin dual. Thus, much of the complicated machinery in
the paper is to obtain the main result for all, and not necessairly Abelian, groups.

As our second main contribution, in~\Cref{sec:algebraic-approach} we relate our results to a recent
theory of Barto et al.~\cite{Barto24:lics}, who developed a systematic approach
to study (in)approximability of promise CSPs, which includes approximability of
promise linear equations, from the viewpoint of universal algebra. In
particular, we show that the proof of our main result (\Cref{th:main}) implies that the collection
of symmetries\footnote{called the \emph{valued minion of plurimorphisms}
in~\cite{Barto24:lics}.} of 
Promise $\eq$ 
can be mapped
homomorphically to the collection of symmetries of Gap Label Cover, a condition
that, based on the algebraic theory from~\cite{Barto24:lics}, is known to guarantee NP-hardness of the former problem.
Thus, our result fits within the algebraic theory of Barto et
al.~\cite{Barto24:lics}.

\paragraph{Related work}
First, extending the work from~\cite{Hastad01:jacm}, Austrin, Brown-Cohen, and
H{\aa}stad established optimal inapproximability of $\eq$ over Abelian groups
with a universal factor graph~\cite{Austrin23:talg}. Similarly, Bhangale and
Stankovic established optimal inapproximability of $\eq$ over non-Abelian groups
with a universal factor graph~\cite{Bhangale23:algorithmica}.
Second, unlike over Abelian groups, for $\eq$ over non-Abelian groups finding a
satisfying assignment is NP-hard even under the promise that one exists. There
is a folklore randomised algorithm for satisfiable $\eq$ instances over
non-Abelian groups (whose approximation factor depends on the group $\gr$ and is
$1/|\gr|$ if $\gr$ is a so-called perfect group but can beat the naive random assignment for non-perfect groups).
Bhangale and Khot showed that this algorithm is optimal~\cite{Bhangale21:stoc},
spawning a number of follow-up works going beyond $\eq$, e.g.~\cite{Bhangale22:stoc,Bhangale23:stocII,Bhangale23:stocIII}.
Third, going beyond $\eq$, building on a long line of work Chan
established optimal (up to a constant factor) NP-hardness for
CSPs~\cite{Chan16:jacm}. There are other works on various inapproximability notions for
CSPs, e.g.,~\cite{Austrin13:toct,Khot14:stoc,Khot14:icalp}.
Finally, we mention that Khot's influential Unique Games
Conjecture~\cite{Khot02stoc} postulates, in one of its equivalent forms,
NP-hardness of finding a barely satisfying solution to a $\eqqq$ instance given
that an almost-satisfying assignment exists (for a large enough domain size).

\paragraph{Paper structure}
The structure of the paper is as follows. 
In~\Cref{sec:background} we set the notation and present the necessary technical background on Fourier analysis over non-Abelian groups.
\Cref{sec:results} presents the main results with discussion and also provides a sketch of the main proof: In~\Cref{sec:reduction} we present the reduction from Gap Label
Cover to Promise $\eq$,
and in
\Cref{sec:outline} we give an overview of the techniques used in the analysis of
this reduction and of the main challenges that arise in extending previous work
to the promise setting. 
\Cref{sec:proof-main} then gives all technical details of the proof of the main result, with the completeness analysis in~\Cref{sec:completeness} and the soundness analysis in~\Cref{sec:soundness}.
Finally, in~\Cref{sec:algebraic-approach} we show that our results fit within the theory of Barto et al.~\cite{Barto24:lics}.

\section{Background}\label{sec:background}

This section gives the necessary background needed in the rest of the paper. We state several results related to Fourier Analysis over non-Abelian groups
and direct products. All proofs can be found in~\Cref{ap:fourier-proofs}. \par

We use $\llbracket \cdot \rrbracket$ to denote the Iverson bracket; i.e.,
$\llbracket P\rrbracket$ is 1 if $P$ is true and $0$ otherwise.
As usual, $[n]$ denotes the set $\{1,2,\ldots,n\}$. \par

\paragraph{Groups and group homomorphisms} A \emph{group} consist of a finite set $\gr$ equipped with a binary associative operation $\cdot\,$, a distinct element $\groupid \in \gr$ which acts as the identity of $\cdot\,$, and where each element $g \in\gr$ has an inverse under $\cdot\,$, denoted by $g^{-1}$. When there is no ambiguity, we denote the group product $g\cdot h$ simply by $gh$.
A group $\gr$ is 
\emph{Abelian} if $gh=hg$ for every $g,h\in\gr$.
A subset $\sgr \subseteq \gr$ of a group $\gr$ is called a \emph{subgroup} of $\gr$, denoted by $\sgr \leq \gr$, if $\sgr$ equipped with the group operation of $\gr$ forms a group. 
Given a group $\gr$, a subgroup $\sgr$ of $\gr$, and an element $g \in\gr$, the \emph{right coset of $\sgr$ in $\gr$ by $g$} is the set $\sgr g :=\{h g \mid h \in \sgr\}$.
The set of right cosets of $\sgr$ in $\gr$ is denoted $\sgr \backslash \gr$.
Let $N$ be a finite set. The $N^{\text{th}}$ \emph{direct power} of $\gr$, denoted by $\gr^N$, is the group whose elements are $N$-tuples $\bg \in \gr^N$ of elements from $\gr$, and where the group operation is taken component-wise, i.e., $\bg \cdot \bh (n)= \bg(n) \cdot \bh(n)$ for each $n \in N$. If $\sgr \leq \gr$, we define $(h \cdot \bg)(n) = h \cdot \bg(n)$ for each $h \in \sgr$ and $\bg \in \gr^N$. With this notation, the notion of coset extends to include the right cosets of $\sgr$ in $\gr^N$ in a natural way. 

A \emph{homomorphism} from a group $\gr_1$ to a group $\gr_2$ is a map $\varphi:\gr_1 \to \gr_2$ which satisfies that $\varphi(g \cdot h)=\varphi(g)\cdot \varphi(h)$ for every $g,h\in\gr_1$.
The domain and image of $\varphi$ are denoted $\dom(\varphi)$ and $\im(\varphi)$ respectively.
Let $N$ be a finite set, $\gr_i$ groups, $i \in [2]$, $\sgr_i \leq \gr_i$, and $\varphi:\sgr_1\to\sgr_2$ be a homomorphism. We say that a function $f:\gr_1^N \to \gr_2$ is \emph{folded over $\varphi$} if
    $f(h\bg)=\varphi(h)f(\bg)$ for all $h\in \sgr_1$ and $\bg\in \gr_1^N$. Given an arbitrary function $f:\gr_1^N \to \gr_2$ and a homomorphism between subgroups, 
    there is a natural way to construct a folded function that resembles $f$. Fix an arbitrary representative, from each right coset of $\sgr_1$ in $\gr_1^N$. For each $\bg\in \gr^N$, denote by $\bg^\dagger$ the representative of $\sgr_1 \bg$, and let $h_\bg \in \sgr_1$ be such that 
    $\bg^\dagger = h_\bg \bg$. Then the \emph{folding of $f$ over $\varphi$} (with respect to this choice of representatives)
    is the map $f_\varphi: \gr_1^N \to \gr_2$
    given by $f_\varphi(\bg) = \varphi(h_\bg^{-1})f(\bg^\dagger)$.

\paragraph{Label Cover} Fix a pair of disjoint finite sets $D$, $E$, called the \emph{label sets}, and a subset $\Pi \subseteq E^D$ of \emph{labeling functions}. An instance of the \emph{Label Cover} problem is a bipartite graph with vertex set $U \sqcup V$ and a labeling function $\pi_{uv} \in \Pi$ for each edge $\{u,v\}$ in the graph.  The task is to decide whether there is a pair of assignments $h_D:U\to D$, $h_E:V\to E$ that satisfies all the constraints, i.e., such that $\pi_{uv}(h_D(u))=h_E(v)$ for each edge $\{u,v\}$.

Given additionally a pair of rational constants $0 < s \leq c \leq 1$, the gap version of this problem, known as the \emph{Gap Label Cover} problem with completeness $c$ and soundness $s$ and denoted $\glc_{D,E}(c,s)$, is the problem of distinguishing instances where a $c$-fraction of the constraints can be satisfied from instances where not even an $s$-fraction of the constraints can be satisfied. 

The hardness of Gap Label Cover with perfect completeness stated below is a consequence of the PCP theorem~\cite{Arora98:jacm-proof,Arora98:jacm-probabilistic} and the Parallel Repetition Theorem~\cite{Raz98}.

\begin{theorem}
\label{th:gap-label-cover}
    For every $\alpha >0$ there exist finite sets $D$, $E$ such that $\glc_{D,E}(1,\alpha)$ is NP-hard.
\end{theorem}

\paragraph{Linear Algebra}
We consider matrices with generalised index sets for convenience. 
Given two finite sets $N$ and $M$, an $N\times M$ complex matrix $A$ consists of a family of complex numbers $A_{i,j}$ indexed by $i\in N$, $j\in M$. Matrix operations are defined in the usual way. 
\par
Given a matrix $A$, we write $\tsp{A}$
to denote its transpose, $\overline{A}$ for its complex conjugate, and $A^*$ for its Hermitian (i.e., its conjugate transpose).
A $N\times M$ matrix is called square if $N=M$. A square matrix $A$ is called \textit{Hermitian} if $A=A^*$, and is called \textit{unitary} if its Hermitian is its inverse. The \textit{trace} of a matrix $A$, denoted $\tr(A)$, is the sum of its diagonal entries.  We denote by $I_N$ the $N\times N$ identity matrix. 
\par

Let $A$ be an $N_1\times N_2$ complex matrix and $B$ be an $M_1\times M_2$ complex matrix. The \emph{tensor product} $A\otimes B$ is a $(N_1\times M_1) \times (N_2\times M_2)$ matrix, where
$(A\otimes B)_{(i,s)(j,t)}=A_{i,j}B_{s,t}$
for each $i\in N_1,j\in N_2, s\in M_1, t\in M_2$. 
We make use of the following linear algebra facts.

\begin{lemma}
    Let $A,B$ be $N\times N$ complex matrices and $C,D$ be $M\times M$ complex matrices. Then the following hold.
    \begin{enumerate}
        \item       $\tr(A)\tr(C)=\tr(A\otimes C)$.
        \item 
        $(AB) \otimes (CD) = 
        (A\otimes C) (B \otimes D)$.
    \end{enumerate}
\end{lemma}

The group of invertible $N
\times N$ complex matrices (equipped with matrix multiplication and matrix inversion) is denoted by $\GL(N)$, and the set of $N\times M$
complex matrices is denoted by $\cc^{N\times M}$.

\subsection{Fourier Analysis over non-Abelian Groups}
Most of the results and definitions in this subsection can be found in~\cite{Terras_1999}. \par
A \emph{representation} of a group $\gr$  is a group homomorphism 
$\gamma: \gr \to \GL(N_\gamma)$ for some finite set $N_\gamma$. We call $|N_\gamma|$ the \emph{dimension} of $\gamma$ and write $\dim_\gamma=|N_\gamma|$.
Given a pair of indices $i,j\in N_\gamma^2$, $\gamma_{i,j}$ denotes the $(i,j)$-th entry of $\gamma$. 
The \emph{character} of a representation $\gamma$, denoted by $\chi_\gamma$, is its trace. The \emph{trivial representation}, denoted $\triv$, maps all group elements to the number one (i.e., the one-dimensional identity matrix). 
\par

A representation $\gamma$ is said to be \emph{unitary} if its image contains only unitary matrices. That is, if 
\[
   \gamma(g)\gamma^*(g) = I_{N_\gamma} \text{ for all $g\in \gr$.}
\]

We say that a function $B: \gr \to \mathbb{C}^{N \times N}$ is \emph{skew-symmetric} (or \emph{skew-Hermitian}) if $B(g^{-1})=B(g)^*$ for all $g\in\gr$. Note that, by definition, a representation is unitary if and only if it is skew-symmetric.

Let $\gr$ be a group and let $\alpha$ and $\beta$ be representations of $\gr$.
We say that $\alpha$ and $\beta$ are equivalent, written $\alpha \simeq \beta$, if
there is an invertible $N_\beta \times N_\alpha$ complex matrix $T$ such that
$\alpha(g)= T^{-1} \beta(g) T$ for all $g\in \gr$. In particular, $\dim_\alpha=\dim_\beta$.  The representation $\beta$ is said to be a \emph{sub-representation} of $\alpha$ if there is
an invertible matrix  $T\in N_\alpha\times M$, where the index set $M$ is a disjoint union of the form $N_\beta\sqcup N$, such that 
$T^{-1} \alpha(g) T$ admits the following block form for all $g\in \gr$:
\[
\begin{pmatrix}
\beta(g) & * \\
    0 & * 
\end{pmatrix}.
\]
The representation $\alpha$ is said to be \emph{irreducible} if all its sub-representations are equivalent to itself. \par

\begin{definition}
    A complete set $\widehat{\gr}$ of inequivalent irreducible unitary representations of a group $\gr$ is a set of irreducible unitary representations of $\gr$ that are pairwise inequivalent and satisfy that any irreducible representation of $\gr$ is equivalent to a representation in $\widehat{\gr}$.    
\end{definition}

From now on, given a group $\gr$, we use $\widehat{\gr}$ to denote some arbitrary and fixed complete set of inequivalent irreducible unitary representations of $\gr$. Note that such a set exists by, e.g., \cite[Proposition 1]{Terras_1999}.\par

We define the space $\mathcal{L}^2(\gr)$ as the vector space of complex-valued
functions over $\gr$, equipped with the following inner product:
\[
\langle F, H \rangle = 
\frac{1}{|\gr|} \sum_{g\in \gr}
F(g) \overline{H(g)}
\]

Let $\gr$ be a group, and let $F:\gr \rightarrow \cc$ be a complex-valued function. Given $\gamma\in \widehat{\gr}$ and $i,j\in N_\gamma$, the \emph{Fourier coefficient} $\widehat{F}(\gamma_{i,j})$ is defined as the product
$\langle F, \gamma_{i,j} \rangle$. 
The matrix entries of the representations $\gamma\in \widehat{\gr}$ form an orthogonal basis of $\mathcal{L}^2(\gr)$, and allow us to perform Fourier analysis on this space, as stated in the following theorem~\cite[Theorem 2]{Terras_1999}.

\begin{theorem}
\label{th:orthogonal_representations}
    Let $\gr$ be a finite group. Then the set 
    \[
    \{
    \gamma_{i,j} \mid \gamma\in \widehat{\gr}, \ i,j \in N_\gamma
    \} \]
    is an orthogonal basis of $\mathcal{L}^2(\gr)$, and  $\dim_\gamma \lVert \gamma_{i,j} \rVert^2=1$ for all $\gamma_{i,j}$. Moreover, the following hold:
    \begin{enumerate}
    \item \emph{Plancherel's Theorem:} Given $F\in \mathcal{L}^2(\gr)$,
    \[
    \lVert F \rVert^2 = 
    \sum_{\gamma\in \widehat{\gr}, i,j\in 
    N_\gamma} \dim_\gamma \vert\widehat{F}(\gamma_{i,j}) \vert^2.
    \]
    \item \emph{Fourier Inversion:} Given $F\in \mathcal{L}^2(\gr)$,
    \[
    F(g) = \sum_{\gamma\in \widehat{\gr}, i,j\in 
    N_\gamma} \dim_\gamma \widehat{F}(\gamma_{i,j}) \gamma_{i,j}(g) \qquad  \text{ for all }g\in \gr.
    \]
\end{enumerate}
\end{theorem}

In particular, for each non-trivial irreducible representation $\gamma$, and each pair of indices $i,j$, the map $\gamma_{i,j}$ is orthogonal to the trivial representation, yielding the following:
\begin{corollary} \label{cor:group-sum-zero}
         $\sum_{g \in \gr} \gamma_{i,j}(g) = |\gr| \langle \gamma_{i,j}, \triv \rangle  = 0$ for each non-trivial $\gamma \in \widehat{\gr}$ and each $i,j \in N_\gamma$.
\end{corollary}

We also consider Fourier transforms of matrix-valued functions $F:\gr \rightarrow \cc^{N_F\times N_F}$. Given $\gamma\in \widehat{\gr}$ and indices $i,j\in N_\gamma$, we define the $N_F\times N_F$ matrix $\widehat{F}(\gamma_{i,j})$ as the one whose $(s,t)$-th entry is
$\widehat{F_{s,t}}(\gamma_{i,j})$ for each $s,t\in N_F$. In other words,
\[
\widehat{F}(\gamma_{i,j})= \frac{1}{|\gr|} \sum_{g\in \gr} F(g) \overline{\gamma_{i,j}(g)}.
\]

The following lemma gives an alternative expression for the contribution of a representation $\gamma \in \widehat{\gr}$
to the Fourier series of a matrix-valued function $F$.

\begin{restatable}{lemma}{LEfouriercoefficientaux}\label{le:fourier_coefficient_aux}
    Let $F:\gr \rightarrow \cc^{N\times N}$ be a map, $\gamma\in \widehat{\gr}$,   
    and $g\in \gr$. Then
    \[
    \sum_{i,j\in N_\gamma} \widehat{F}(\gamma_{i,j}) \gamma_{i,j}(g) =
    \frac{1}{|\gr|} \sum_{h\in \gr} F(h) \chi_\gamma(h^{-1}g).
    \]
\end{restatable}

The following is a well-known fact about sums of characters of irreducible representations. 
\begin{lemma}[\protect{\cite[Lemma 2]{Terras_1999}}]
\label{le:sum-dim-char}
    Let $\gr$ be a finite group. Then 
        \[
        \sum_{\gamma\in \widehat{G}} \dim_\gamma \chi_\gamma(g) = \begin{cases}
            |\gr| \text{ if $g=\groupid_\gr$, and}\\
            0 \text{ otherwise.}
        \end{cases}
\]
In particular, using that $\chi_\gamma(\groupid_\gr)=\dim_\gamma$ we obtain that
$\sum_{\gamma \in \widehat{\gr}}
        \dim_\gamma^2 = |\gr|$.
\end{lemma}

Given representations $\rho_1,\dots, \rho_k$, and non-negative integers $n_1,\dots, n_k$, we write $n_1\rho_1 \oplus \dots \oplus n_k \rho_k$ for the representation $\gamma$ of $\gr$ satisfying that for each $g\in \gr$, $\gamma(g)$ is a block-diagonal matrix whose $i$th diagonal block is itself a block-diagonal matrix consisting of $n_i$ blocks of $\rho_i(g)$.

\begin{definition}
    Let $\gr$ be a group, and let $\gamma$ be a representation of $\gr$. We say that a non-negative integer $n$ is the multiplicity of $\rho\in \widehat{\gr}$ in $\gamma$
    if $\gamma \simeq n\rho \oplus \gamma^\prime$, where 
    $\gamma^\prime$ is another representation of $\gr$ that does not have $\rho$ as a sub-representation.
\end{definition}

Although it is not obvious from the definition above, the multiplicity of an irreducible representation in another representation is a well-defined non-negative integer. In fact, the following facts hold.

\begin{lemma}[\protect{\cite[Proposition 2]{Terras_1999}}]
\label{le:complete_reducibility}
    Let $\gamma$ be a representation of $\gr$.
    Then
    $\gamma \simeq \bigoplus_{\rho\in \widehat{\gr}} n_\rho \rho$, where $n_\rho$ denotes the multiplicity of $\rho$ in $\gamma$ for each $\rho \in \widehat{\gr}$.
\end{lemma}

\begin{lemma}[\protect{\cite[Proposition 3]{Terras_1999}}]
\label{le:multiplicity}
    Let $\gamma$ be a representation of $\gr$, and let $\rho\in \widehat{\gr}$. Then the multiplicity of $\rho$ in $\gamma$ equals both
    $\langle
   \chi_\rho, \chi_\gamma
    \rangle$ and $\langle \chi_\gamma, \chi_\rho \rangle$.
\end{lemma}

We define the \emph{right-regular representation of $\gr$ on $\sgr \backslash \gr$},  denoted $R_{\gr,\sgr}$, 
to be the representation of $\gr$ with
$N_{R_{\gr,\sgr}}=\sgr \backslash \gr$ whose entries are defined as follows. Given two cosets $\sgr g_1, \sgr g_2$ and an element $g\in \gr$, the $(\sgr g_1, \sgr g_2)$-entry of $R_{\gr,\sgr}(g)$ is $1$ if $\sgr g_1 = \sgr g_2 g^{-1}$, and zero otherwise. In other words, $R_{\gr,\sgr}(g)$ is the matrix describing the permutation on the cosets $\sgr \backslash \gr$ resulting from right-multiplication with $g^{-1}$. The representation $R_{\gr,\sgr}$ is the representation induced on $\gr$ by the trivial representation of $\sgr$. When $\sgr$ is the trivial subgroup containing just the identity element, we simply write $R_\gr$ for $R_{\gr,\sgr}$. The representation $R_\gr$ is known as the \emph{right-regular representation of $\gr$}.
An application of the Frobenius Reciprocity Law yields the following result. 

\begin{restatable}{lemma}{LEmultiplicitytrivialrestriction}
\label{le:multiplicity_trivial_restriction}
    Let $\sgr\leq \gr$ be groups, and let $\rho\in \widehat{\gr}$. Then the multiplicity of $\rho$ in $R_{\gr,\sgr}$ is the same as the multiplicity of the trivial representation in $\rho\vert_\sgr$.
\end{restatable}

The following result appears in~\cite[Lemma 2]{Terras_1999} for the so-called left-regular representation, 
but the same proof also works for our modified statement. 

\begin{lemma}
\label{le:dimension}
    Let $\gr$ be a group and let $R_\gr$ be the right-regular representation of $\gr$. Then $R_{\gr} \simeq \bigoplus_{\rho\in \widehat{\gr}} \dim_\rho \rho$.
\end{lemma}

We will also need the following result.

\begin{restatable}{lemma}{LEsubrepresentationsrestriction}
\label{le:subrepresentations_restriction}
Let $\sgr\leq \gr$ be groups. Then
\[
\sum_{\rho\in \widehat{\gr}}
\dim_\rho \langle \chi_\rho, \chi_{R_{\gr,\sgr}} \rangle = 
\frac{|\gr|}{|\sgr|}.
\]
\end{restatable}

Let $N$ be a finite set. Given a pair of functions function $F,H: \gr \to \mathbb{C}^{N \times N}$, we define their \emph{convolution} $F * H$ by 
\[(F * H)(g) := \frac{1}{|\gr|} \sum_{h \in \gr} F(h) H(h^{-1} g).\]
The following observation will be useful.
\begin{restatable}{lemma}{LEfourierconvolution}
\label{lem:fourier_convolution}
    Let  $F,G: \gr \to \cc^{N\times N}$, $\gamma \in \widehat{\gr}$. Let
    $x,y\in N$,
    $i,j\in N_\gamma$. Then, \[\widehat{(F * G)}_{x,y}(\gamma_{i,j}) =
    \sum_{z\in N} \sum_{k\in N_\gamma}
    \widehat{F_{x,z}}(\gamma_{i,k})\widehat{G_{z,y}}(\gamma_{k,j}).\]
\end{restatable}

\subsubsection{Fourier Analysis over Direct Products}

\begin{theorem}[\cite{Terras_1999}]
    Let $\gr,\sgr$ be finite groups,
    and let $\widehat{\gr}$ and $\widehat{\sgr}$ be complete sets of inequivalent unitary irreducible representations of $\gr$ and $\sgr$ respectively. Then the set
    $\{
    \alpha \otimes \beta \vert \alpha\in \widehat{\gr}, \beta\in  \widehat{\sgr}
    \}$ is a complete set of inequivalent unitary irreducible representations over $\gr\times \sgr$.
\end{theorem}
This result allows us to identify $\widehat{\gr^D}$ with $(\widehat{\gr})^D$, for a given group $\gr$ and finite set $D$. This way, an element $\rho\in \widehat{\gr^D}$ is given by a tuple $(\rho^d)_{d\in D}$ where $\rho^d\in \widehat{\gr}$ for each $d\in D$ in such a way that
\[
\rho(\bg)= \bigotimes_{d\in D}
\rho^d(\bg(d)),
\]
for all $\bg \in \gr^D$. Observe that we use superscripts for the ``components'' of the representation $\rho$ on the power group $\gr^D$, rather than subscripts, which we utilise to denote matrix entries. 
The \emph{degree} of a representation $\rho\in \widehat{\gr^D}$, written $|\rho|$, is the number of indices $d\in D$ for
which $\rho^d$ is non-trivial.\footnote{This quantity is called ``weight''
in~\cite{EHR04:tcs,Bhangale23:algorithmica}.} \par
Let $\gr$ be a group, $E, D$ finite sets, $\pi: D\rightarrow E$ a map, and $\tau\in \widehat{\gr^E}, \rho\in \widehat{\gr^D}$ representations. We define the unitary representation $\rho^\pi$
of $\gr^E$ by 
\[\rho^\pi( \bg)=  
\bigotimes_{e\in E}
\bigotimes_{d\in \pi^{-1}(e)}
\rho^{d}(\bg(e)).
\] for all $\bg \in \gr^E$.
Observe that $N_{\rho^\pi}=N_\rho= \prod_{d\in D} N_{\rho^d}$. Given indices $i=(i_d)_{d\in D}, j=(j_d)_{d\in D}$, we have 
\[\rho_{i,j}^\pi( \bg)=  
\bigotimes_{e\in E}
\bigotimes_{d\in \pi^{-1}(e)}
\rho_{i_d,j_d}^{d}(\bg(e)).
\] 

We write $\tau \sim_\pi \rho$ if for each $e\in E$ for which $\tau^e$ is non-trivial, there is at least one $d\in \pi^{-1}(e)$ for which $\rho^d$ is non-trivial. 

\begin{restatable}{lemma}{LEsimilarrepresentations}
\label{le:similar_representations}
    Let $\gr$ be a group, $E, D$ finite sets, $\pi: D\rightarrow E$ a map, and $\tau\in \widehat{\gr^E}, \rho\in \widehat{\gr^D}$ representations. Then 
    \begin{enumerate}
        \item $\tau \sim_\pi \rho$ implies that $|\tau| \leq |\rho|$.
        \item $\tau \not\sim_\pi \rho$ implies that 
        $\langle \tau_{s,t}, \rho^\pi_{i,j} \rangle_{\gr^E}=0$
        for all $s,t\in N_\tau$, $i,j\in N_\rho$.
    \end{enumerate}
\end{restatable}

In our proofs, we analyse complex-valued functions $H$ over direct products $\gr^D$ of the form 
$H(\tuple{g})= \ex{\bm{\nu}}[F(\tuple{g}\cdot \bm{\nu})]$, where $\bm{\nu}$ is some ``random noise tuple''. The following result relates the Fourier coefficients of $F$ and $H$ in this setting.

\begin{restatable}{lemma}{LEepsilonnoise}
  \label{le:epsilon-noise}
    Let $\gr$ be a finite group, $D$ a finite set, and $0<\epsilon<1$ a real number. Let $\bm{\nu}=(\nu_d)_{d\in D}$ be a random element from $\gr^D$ chosen as follows. Independently for each $d\in D$, the element $\nu_d$ equals $\groupid_\gr$ with probability $1-\epsilon$, and is picked uniformly at random from $\gr$ otherwise.     
    Let $F:\gr^D\to\mathbb{C}$ and define $H(\ba):= \ex{\bm{\nu}}[F(\ba\cdot \bm{\nu})]$. Then, for each $\rho\in \widehat{\gr^D}$ and $i,j\in N_\rho$, we have
    \[
    \widehat{H}(\rho_{i,j})= 
    (1-\epsilon)^{|\rho|} \widehat{F}(\rho_{i,j}).
    \] 
\end{restatable}

\section{Overview of Results}\label{sec:results}

Let $\gr_1,\gr_2$ be two groups and $\varphi$ a group homomorphism with domain
$\dom(\varphi) \leq \gr_1$ and image $\im(\varphi) \leq \gr_2$ that extends to a full
homomorphism from $\gr_1$ to $\gr_2$.
We shall refer to triples $(\gr_1,\gr_2, \varphi)$ of this kind as \emph{templates}. Further, let $0 < s \leq c \leq 1$ be rational constants. 
We consider the problem $\eq(\gr_1, \gr_2, \varphi,c,s)$ which asks, given a
 weighted system of linear equations with exactly three variables in each equation and
constants in $\dom(\varphi)$ that is $c$-satisfiable in $\gr_1$, to decide whether
there exists an $s$-approximation in $\gr_2$, where the constants are
interpreted through $\varphi$. 

To be more precise, an instance to  $\eq(\gr_1, \gr_2, \varphi,c,s)$ over a set of variables $X$ is a weighted systems of linear equations where each equation is of the form \[x^i y^j z^k = g\] for some $x,y,z \in X$, $g \in \dom(\varphi)$,  $i,j,k \in \{-1,1\}$, and each equation has a non-negative rational weight. Without loss of generality, we assume that the weights are normalised, i.e., sum up to 1.
For $t \in [2]$, an assignment $f:X\to \gr_t$ \emph{satisfies} an equation  $x^i y^j z^k = g$ in $\gr_t$ if $f(x)^i f(y)^j f(z)^k = g$ for $t=1$, and $f(x)^i f(y)^j f(z)^k = \varphi(g)$ for $t=2$.
The task then is to accept if there is an assignment that satisfies a $c$-fraction (i.e., a fraction of total weight $c$) of equations in $\gr_1$, and to reject if there is no assignment that satisfies an $s$-fraction of the equations in $\gr_2$. It is easy to verify that, if $(\gr_1,\gr_2, \varphi)$ is a template and $s \leq c$, then the sets of accept and reject instances are, in fact, disjoint.\footnote{$\eq$ can be alternatively phrased as a promise constraint satisfaction problem, cf.~\Cref{sec:algebraic-approach} for details.}

$\eq(\gr_1,\gr_2,\varphi,c,s)$ is trivially tractable when $\im(\varphi)=\{\groupid\}$, so we focus on the case where $|\im(\varphi)|\geq 2$.
The main result of this paper is that $\eq(\gr_1,\gr_2,\varphi,1-\epsilon,1/|\im(\varphi)|+\delta)$ is NP-hard for all $\epsilon, \delta>0$ for which the problem is well-defined. This is achieved by a reduction from the Gap Label Cover problem with perfect completeness and soundness $\alpha=\delta^2/(4\kappa|\gr_1|^\kappa|\gr_2|^{4})$, where $\kappa=\lceil(\log_2\delta -2)/(\log_2(1-\epsilon))\rceil$. 
\begin{theorem}[\textbf{Main}]
\label{th:main}
     Let $\epsilon, \delta$ be positive constants satisfying $1-\epsilon
     \geq1/|\im(\varphi)|+\delta$. Then, $\eq(\gr_1,\gr_2,\varphi,1-\epsilon,1/|\im(\varphi)|+\delta)$ is NP-hard.
\end{theorem}
The hardness result in~\Cref{th:main} is tight for many, but perhaps surprisingly not all, templates. 
We call a template $(\gr_1,\gr_2, \varphi)$
\emph{cubic} if for every $h\in \im(\varphi)$ there is an element $g\in \gr_2$ satisfying $g^3=h$.
\Cref{th:main} is tight for cubic templates. Indeed, for these templates, the random assignment over $\im(\varphi)$ achieves a $1/|\im(\varphi)|$ expected
fraction of satisfied equations (and this can be derandomised, e.g., by the
method of conditional expectations). \par

\begin{theorem}
\label{th:main-cubic}
    Let $(\gr_1, \gr_2, \varphi)$ be a cubic template 
    and $0<s\leq c<1$. Then, $\eq(\gr_1, \gr_2, \varphi, c, s)$ is tractable if $s \leq 1/|\im(\varphi)|$ and NP-hard otherwise.
\end{theorem}

Let us now turn to non-cubic templates. 
An equation is \emph{unsatisfiable} if it is
of the form $x^3=h$ or $x^{-3}=h$ for some $h\in \dom(\varphi)$ such that $g^3 \neq \varphi(h)$ for all $g \in \gr_2$.\footnote{Note
that, since $\varphi$ extends to a homomorphism from $\gr_1$ to $\gr_2$, this also implies
that $g^3 \neq h$ for all $g \in \gr_1$.}
Note that a
template has unsatisfiable equations if and only if it is non-cubic. Note that
the naive random assignment cannot achieve a positive approximation factor in
systems of equations over non-cubic templates since 
the system could consist exclusively of unsatisfiable equations. However, there is a simple
algorithm for $\eq(\gr_1,\gr_2,\varphi,c,c/|\im(\varphi)|)$ that works even for
non-cubic templates, which we describe next.

Given a weighted
system of equations over $(\gr_1,\gr_2,\varphi)$, consider its set of unsatisfiable equations. 
Since $\varphi$ extends to a full homomorphism, if the total weight of the set of unsatisfiable equations is more than $1-c$, then the instance cannot be $c$-satisfiable in $\gr_1$, hence, reject.
Otherwise, the random assignment over $\im(\varphi)$ satisfies at least a
$1/|\im(\varphi)|$-fraction of the satisfiable equations over $\gr_2$, which is at least a $c/|\im(\varphi)|$-fraction of the entire system. 
It is a simple corollary of~\Cref{th:main} that this algorithm is optimal for non-cubic groups, leading to the following result. Details are deferred to~\Cref{ap:non-cubic}.

\begin{restatable}{theorem}{THnoncubic}
\label{th:main-non-cubic}
    Let $(\gr_1, \gr_2, \varphi)$ be a non-cubic template 
    and $0<s\leq c<1$. Then, $\eq(\gr_1,$ $ \gr_2, \varphi, c, s)$ is tractable if $s/c \leq 1/|\im(\varphi)|$ and NP-hard otherwise.
\end{restatable}

\subsection{Reduction} \label{sec:reduction}

For the rest of the section we outline the proof of our main result,~\Cref{th:main}. From now on we fix a template $(\gr_1, \gr_2, \varphi)$, and positive constants $\delta, \epsilon>0$ with $1/|\im(\varphi)| + \delta \leq 1 - \epsilon$. We define $\sgr_1 =\dom(\varphi) \leq \gr_1$ and $\sgr_2 = \im(\varphi) \leq \gr_2$.

Our proof follows from a reduction from 
$\glc_{D,E}(1,\alpha)$ where 
\[\alpha=\frac{\delta^2}{4\kappa|\gr_1|^\kappa|\gr_2|^{4}}, \quad \kappa=\Bigg\lceil\frac{\log_2\delta -2}{\log_2(1-\epsilon)}\Bigg\rceil,\] and $D,E$ are
chosen to be large enough so that $\glc_{D,E}(1,\alpha)$ is NP-hard by the PCP
theorem~\cite{Arora98:jacm-proof,Arora98:jacm-probabilistic,Raz98}
(cf.~\Cref{th:gap-label-cover}). This reduction constructs an instance $\Phi_\Sigma$ of $\eq(\gr_1, \gr_2,$ $\varphi, 1-\epsilon, 1/|\sgr_2| + \delta)$ for any given instance $\Sigma$ of Gap Label Cover as described below. 

Let $U \sqcup V$ be the underlying vertex set of $\Sigma$, $D, E$ be the disjoint sets of labels, and $\pi_{uv}$ be the labeling functions. 
We fix representatives from each right coset in $\sgr_1\backslash\gr_1^D$ and
$\sgr_1 \backslash \gr_1^E$. Given a tuple $\tuple{x}$ in either $\gr_1^D$ or $\gr_1^E$ we write $\bx^\dagger$ for the representative of the coset $\sgr_1\bx$. 
Let $X= \{
u_\tuple{b} \, \vert u\in U, \tuple{b}\in \gr_1^D 
\} \sqcup \{
v_\tuple{a} \, \vert v\in V, \tuple{a}\in \gr_1^E 
\}$. Then $\Phi_\Sigma$ is the weighted system of equations over $X$ that contains the equation  
\begin{equation}
\label{eq:reduction-def}
v_{\tuple{a}^\dagger} u_{\tuple{b}^{s_1}}^{s_1} u_{\tuple{c}^{s_2}}^{s_2} = g_{\tuple{a}}
\end{equation}
for each edge $\{u,v\}$ of $\Sigma$, $\tuple{a}\in \gr_1^E$, $\tuple{b} \in \gr_1^D$,
$s_1,s_2\in \{-1,1\}$, where  $\tuple{c}$ stands for
$\tuple{b}^{-1} (\tuple{a} \circ \pi_{uv})^{-1}\bm{\nu}$ and  $\bm{\nu} \in \gr_1^D$ is a small perturbation factor. The element $g_\tuple{a}$
is chosen so that $\tuple{a}^\dagger= g_\tuple{a} \tuple{a}$. The weight of this equation in $\Phi_\Sigma$ is the joint probability of the independent events described in~\Cref{fig:probs-reduction}.

\begin{figure}[ht]
    \centering
\begin{align*}
\boxed{
\begin{array}{llllll}
\mbox{($1$)}\quad & \text{The
edge $\{u,v\}$ is chosen uniformly at random among all edges of $\Sigma$.}
\\[0.5em]
\mbox{($2$)}\quad & 
\text{The elements $\tuple{a}$ and $\tuple{b}$ are chosen uniformly at random from $\gr_1^E$ and $\gr_1^D$ respectively.}
\\[0.5em]
\mbox{($3$)}\quad &
\text{The element $\bm{\nu}\in \gr_1^D$ is chosen so that for each $d\in D$, independently, 
$\bm{\nu}(d)=\groupid_{\gr_1}$ with} \\ & \text{probability $1-\epsilon$, and $\bm{\nu}(d)$ is selected uniformly at random from $\gr_1$  with probability $\epsilon$.}\\[0.5em]
\mbox{($4$)}\quad & \text{The signs
$s_1,s_2$ are chosen uniformly at random from $\{-1,1\}$.}
\end{array}
}
\end{align*}
 \caption{The sampling procedure for $\Phi_\Sigma$.}
 \label{fig:probs-reduction}
\end{figure}

Let us describe assignments of $\Phi_\Sigma$ over $\gr_i$ for $i=1,2$. Formally, an assignment of $\Phi_\Sigma$ over $\gr_i$ is a map $h: X\rightarrow \gr_i$. Such an assignment can be described by two families of maps $A=(A_v)_{v\in V}$ from $\gr^E_1$ to $\gr_i$ and $B=(B_u)_{u\in U}$ from $\gr_1^D$ to $\gr_i$ by letting
$A_v(\tuple{a})= h(v_{\tuple{a}}) $ for all $v\in V, \tuple{a}\in \gr_1^E$, and $B_u(\tuple{b})=h(u_{\tuple{b}})$ for all $u\in U, \tuple{b}\in \gr_1^D$. It will be more convenient to talk about the pair $(A,B)$ rather than the map $h$ itself, so we will write $\Phi^{\gr_i}_\Sigma(
A,B)$ to refer to the proportion of equations satisfied by the assignment $h$. 
Let us give a more useful expression for $\Phi^{\gr_i}_\Sigma(
A,B)$. When $i=1$, we can write
\[
\Phi^{\gr_1}_\Sigma(
A,B) = 
\ex{\substack{uv,\tuple{a},\tuple{b}, \\ \bm{\nu},s_1,s_2}} \left[ 
\llbracket A_v(\tuple{a}^\dagger) B_u(\tuple{b}^{s_1})^{s_1} B_u((\tuple{b}^{-1} (\tuple{a} \circ \pi_{uv})^{-1}\bm{\nu})^{s_2})^{s_2}    = g_{\tuple{a}} \rrbracket 
\right],
\]
where the expectation is taken over the probabilities described in~\Cref{fig:probs-reduction}, and we use $uv$ as a shorthand for an edge $\{u,v\}$. 
Folding the assignments $A_v$ over the identity on $\sgr_1$ and using the fact that $(A_v)_{\mathrm{id}_{\sgr_1}}(\tuple{a})=g_{\tuple{a}}^{-1}A_v(\tuple{a}^\dagger)$, we obtain
\begin{equation}
\label{eq:payoff_gr1}
\Phi^{\gr_1}_\Sigma(
A,B) = 
\ex{\substack{uv,\tuple{a},\tuple{b}, \\ \bm{\nu},s_1,s_2}} \left[ 
\llbracket (A_v)_{\mathrm{id}_{\sgr_1}}(\tuple{a}) B_u(\tuple{b}^{s_1})^{s_1} B_u((\tuple{b}^{-1} (\tuple{a} \circ \pi_{uv})^{-1}\bm{\nu})^{s_2})^{s_2}    = \groupid_{\gr_1} \rrbracket 
\right]. 
\end{equation}
Analogously, when $i=2$
and $A_v, B_u$ are families of maps to $\gr_2$, we obtain a similar expression for $\Phi^{\gr_2}_\Sigma(A,B)$:
\begin{equation}
\label{eq:payoff_gr2}
\Phi^{\gr_2}_\Sigma(
A,B) = 
\ex{\substack{uv,\tuple{a},\tuple{b}, \\ \bm{\nu},s_1,s_2}} \left[ 
\llbracket (A_v)_{\varphi}(\tuple{a}) B_u(\tuple{b}^{s_1})^{s_1} B_u((\tuple{b}^{-1} (\tuple{a} \circ \pi_{uv})^{-1}\bm{\nu})^{s_2})^{s_2}    = \groupid_{\gr_2} \rrbracket 
\right]. 
\end{equation}

That is, a pair of assignments $(A,B)$ satisfies an equation in $\Phi_\Sigma$ if and only if the corresponding pair of assignments obtained by folding $A$ (over $\mathrm{id}_{\sgr_1}$ and $\varphi$ respectively) maps the equation to the group identity (respectively, in $\gr_1$ and $\gr_2$). Thus, folding allows us to focus exclusively on the identity terms in these expectations, which will be useful in the analysis of the reduction.

\Cref{th:main} follows from our completeness and soundness bounds for 
    $\Phi_\Sigma$, stated in the next results, using the fact that
    by~\Cref{th:gap-label-cover}, there are finite sets $D,E$ such that
    $\glc_{D,E}(1,\alpha)$ is NP-hard for the value of $\alpha$ chosen
    in~\Cref{th:soundness} below. The proofs of the completeness and soundness bounds can be found in~\Cref{sec:completeness} and~\Cref{sec:soundness} respectively. 

\begin{theorem}[Completeness]
\label{th:completeness}
    Let $\Sigma$ be a Gap Label Cover instance and $\Phi_\Sigma$ be the system defined in~\eqref{eq:reduction-def}. Suppose that $\Sigma$ is $1$-satisfiable. Then $\Phi_\Sigma$ is $(1 -\epsilon)$-satisfiable in $\gr_1$.
\end{theorem}

\begin{theorem}[Soundness]
\label{th:soundness}
 Let $\Sigma$ be a Gap Label Cover instance and $\Phi_\Sigma$ be the system defined 
 in~\eqref{eq:reduction-def}. Suppose that $\Phi_\Sigma$ is $(1/|\sgr_2| + \delta)$-satisfiable in $\gr_2$. Then $\Sigma$ is $\alpha$-satisfiable, where $\alpha=\delta^2/(4\kappa|\gr_1|^\kappa|\gr_2|^{4})$ and $\kappa=\lceil(\log_2\delta -2)/(\log_2(1-\epsilon)\rceil$. 
 \end{theorem}

\subsection{Proof Outline} \label{sec:outline}

The main difficulty in proving the correctness of our reduction lies in showing the soundness bound (\Cref{th:soundness}). The completeness result (\Cref{th:completeness}) is relatively straightforward and follows as in~\cite{EHR04:tcs}. In summary, suppose the Gap Label Cover instance $\Sigma$ is satisfied by a pair of assignments $h_D: U \rightarrow D$, 
$h_E: V\rightarrow E$. Then we find families $A,B$ such that $\Phi^{\gr_1}_\Sigma(A, B)\geq 1 - \epsilon$ by letting $A_v$ be the $h_E(v)$-th projection and $B_u$
be the $h_D(u)$-th projection for each $v\in V, u\in U$. As usual, the noise introduced by the perturbation factor $\bm{\nu}$ is what forces us to give up perfect completeness.\par
The idea behind our soundness analysis has appeared many times in the literature (e.g., \cite{Hastad01:jacm,EHR04:tcs,Bhangale21:stoc}), but the approach taken in~\cite{EHR04:tcs} is the most similar to ours. 
Suppose that there are assignments $A, B$, satisfying 
\begin{equation}
\label{eq:proof_outline_1}
\Phi_\Sigma^{\gr_2}(A, B)\geq 
\frac{1}{|\sgr_2|} + \delta.
\end{equation}
In view of~\eqref{eq:payoff_gr2}, this inequality
can be understood as a lower bound for the success probability of the following $3$-query dictatorship test: Sample all parameters
according to the distribution shown in~\Cref{fig:probs-reduction}, and then
query the values $(A_v)_\varphi(\tuple{a})$, $B_u(\tuple{b}^{s_1})^{s_1}$, and
$B_u((\tuple{b}^{-1} (\tuple{a} \circ \pi_{uv})^{-1}\bm{\nu})^{s_2})^{s_2}$. The test
is passed if the product of the three values is the group identity, and failed
otherwise. 
The soundness proof consists in showing that \eqref{eq:proof_outline_1}
implies that the functions $(A_v)_{\varphi}: \gr_1^E \rightarrow \gr_2$ 
and $B_u: \gr_1^D \rightarrow \gr_2$ are ``close'' to dictators (i.e.,
projections) for each $v\in V$, $u\in U$. Then, this fact allows us to find a
good solution to the starting Gap Label Cover instance $\Sigma$. Indeed, suppose that for each $v\in V$ the map $(A_v)_{\varphi}$
is the projection on the $e_v$-th coordinate, and for each $u\in U$, the map $B_u$ is the projection on the $d_u$-th coordinate. Then the assignment mapping $v$ to $e_v$ and $u$ to $d_u$ for each $v\in V, u\in U$ is a good solution for $\Sigma$. However, it is not clear how to extend this simple idea to the case where the maps $(A_v)_\varphi, B_u$ are not projections. \par
In order to find a good solution for $\Sigma$ in this general case, we first find suitable maps $\gamma_1, \gamma_2:\gr_2 \rightarrow \cc$ and analyse 
$\gamma_1 \circ (A_v)_\varphi$, $\gamma_2 \circ B_u$. Now, using the fact that $(A_v)_\varphi$
and $B_u$ are close to projections, we can prove that choosing the labels $e, d$ for the vertices $v,u$ according to the ``low-degree influence" of the $e$-th coordinate in $\gamma_1 \circ (A_v)_\varphi$ and the $d$-th coordinate in $\gamma_2 \circ B_u$ yields a good randomised assignment of $\Sigma$. \par
This overview so far also applies to the soundness analysis of~\cite{EHR04:tcs}. Let us give more detail and highlight the main differences that sets our work apart. The first important difference has to do with the choice of $\gamma_1, \gamma_2$. We define
$\gamma_1= \omega_{x,y}$, and $\gamma_2= \omega_{y,z}$, 
where $\omega$ is some irreducible representation of $\gr_2$, and $x,y,z$ are suitable indices in $N_\omega$. In~\cite{EHR04:tcs}, the representation $\omega$ is a non-trivial representation chosen so that
\begin{equation*}
\left|
\ex{}\left[
\chi_\omega\left( (A_v)_{\varphi}(\tuple{a}) B_u(\tuple{b}^{s_1})^{s_1} B_u((\tuple{b}^{-1} (\tuple{a} \circ \pi)^{-1}\bm{\nu})^{s_2})^{s_2} \right)
\right] \right| \geq \dim_\omega \delta.
\end{equation*}
Here the expectation is taken over the probability space described in~\Cref{fig:probs-reduction}, and the dependence of $\pi$ on the edge $\{u,v\}$ is left implicit.
In our case, rather than using the Fourier characters for choosing $\omega$, we consider ``penalized characters'' $\widetilde{\chi_\omega}$. We define $\widetilde{\chi_\omega}: \gr_2 \rightarrow \cc$ as the map 
$\chi_\omega - \eta_\omega$, where
the penalty $\eta_\omega$ is the multiplicity of the trivial representation in the restriction $\omega\vert_{\sgr_2}$. This way, we pick $\omega\in\widehat{\gr_2}$ so that the previous inequality holds after replacing $\chi_\omega$ with $\widetilde{\chi_\omega}$. Equivalently, we find $\omega$ satisfying
\begin{equation}
\label{eq:choice_of_representation3}
\left|
\ex{}\left[
\chi_\omega\left( (A_v)_{\varphi}(\tuple{a}) B_u(\tuple{b}^{s_1})^{s_1} B_u((\tuple{b}^{-1} (\tuple{a} \circ \pi)^{-1}\bm{\nu})^{s_2})^{s_2} \right)
\right] \right| \geq  \dim_\omega \delta+ \eta_\omega .
\end{equation}
The fact that such $\omega$ exists is a consequence of~\eqref{eq:proof_outline_1} together with $\sum_{\omega\in \widehat{\gr_2}} \dim_\omega \eta_\omega = |\gr_2|/|\sgr_2|$, which follows from the Frobenius Reciprocity Theorem, as shown in~\Cref{le:subrepresentations_restriction}. This additional factor of $\eta_\omega$ is crucial to our soundness analysis, as we will see. \par
Define the map
$\A= \omega\circ(A_v)_{\varphi}$ and
the map
$\B: \gr_1^D \rightarrow \gr_2$ given by $\B(\tuple{b}) = \ex{s\in \{-1,1\}} \omega \circ B_u(\tuple{b}^s)^s$, where $s\in \{-1, 1\}$ is distributed uniformly.\footnote{Observe that the maps $\A$ and $\B$ depend on the hidden parameters $v$ and $u$ respectively.} To show the soundness bound we consider the Fourier expansions of $\A$ and $\B*\B$ in the expression 
\[
\left|
\tr \, 
\ex{}\left[
\A(\tuple{a})
(\B*\B)( (\tuple{a} \circ \pi)^{-1}\bm{\nu}) 
\right]
\right|,\] which is just a rearrangement of the left-hand-side in the previous inequality. More precisely, we look at the equivalent expression 
\begin{equation}
\label{eq:choice_of_representation2}
\left|
\tr \, 
\ex{}\left[ 
\left(
\sum_{\tau\in \widehat{\gr^E}, s,t\in N_\tau}
\dim_\tau \widehat{\A}(\tau_{s,t}) \tau_{s,t}(\tuple{a})
\right)
\left(
\sum_{\rho\in \widehat{\gr^D}, i,j\in N_\rho}
\dim_\rho
\widehat{(\B*\B)}(\rho_{i,j})
\rho_{i,j}((\tuple{a} \circ \pi)^{-1}\bm{\nu}) \right)
\right]
\right|.
\end{equation}

 Our goal is to find a bound $\kappa$, independent of $|D|,|E|$, satisfying that the contribution to this expression of 
non-trivial representations $\tau, \rho$ of degree less than $\kappa$
is at least $\dim_\omega \delta/2$. This is achieved by controlling the contribution of the trivial term and the contribution of high-degree terms, as indicated by~\Cref{le:soundness-aux-1} and~\Cref{le:EHR-23} respectively. 
The second main difference of our soundness analysis compared to~\cite{EHR04:tcs} is our handling of the trivial term. In~\Cref{le:soundness-aux-1} we prove that 
\[ \left|
\tr \, 
\ex{}\left[ 
\widehat{\A}(1) 
\left(
\sum_{\rho\in \widehat{\gr^D}, i,j\in N_\rho}
\dim_\rho
\widehat{(\B*\B)}(\rho_{i,j})
\rho_{i,j}((\tuple{a} \circ \pi)^{-1}\bm{\nu}) \right)
\right]
\right| \leq \eta_\omega.
\]

In the non-promise setting, this bound is not necessary. Roughly, under the stronger notion of folding used in~\cite{EHR04:tcs}, it is possible to show that $\widehat{\A}(1)$ vanishes. Our weaker notion of folding does not allow us to prove the same result,\footnote{Our notion of folding can be arbitrarily weak: any function $f:\gr_1^N \rightarrow \gr_2$ is folded over the map $\varphi$ sending the group identity of $\gr_1$ to the identity of $\gr_2$.} but we are still able to leverage folding to obtain the above bound. This mismatch with~\cite{EHR04:tcs} is the reason why the extra $\eta_\omega$ term was required in~\eqref{eq:choice_of_representation3}. The key insight in the proof of~\Cref{le:soundness-aux-1} is that if $F:\gr_1^E \rightarrow \gr_2$ is folded over $\varphi$, then the trace of $\widehat{(\omega \circ F)}(1)$ is at most $\eta_\omega$ in absolute value. \par

Our analysis of high-degree terms is in the same spirit as previous works that show hardness of approximation in the imperfect completeness setting. In~\Cref{le:EHR-23} we prove that
\begin{align*}
& \nonumber  \left| \tr \, \ex{} \left[ \left(
    \sum_{\tau\in \widehat{\gr_1^E}, \tau\neq 1}
    \sum_{s,t\in N_\tau}
\dim_\tau
\widehat{\A}(\tau_{s,t}) \tau_{s,t}(\tuple{a}) \right) \times
\right. \right. 
\\ 
&
\left. \left. 
\left(
\sum_{\rho \in \widehat{\gr_1^D}, |\rho| \geq \kappa}
\sum_{i,j \in N_\rho}
\dim_\rho \widehat{(\B*\B)}(\rho_{i,j}) \rho_{i,j}( (\tuple{a} \circ \pi)^{-1} \bm{\nu})  \right)  \right] \right| \leq (\dim_\omega \delta)/2
\end{align*}
for all $\kappa\geq (\log_2 \delta - 2)/\log_2(1-\epsilon)$. The essential idea is that the ``noise vector'' $\bm{\nu}$ has a smoothing effect that limits the contribution of high-degree terms in~\eqref{eq:choice_of_representation2}. \par

Finally, having established that the contribution of non-trivial terms of degree less than $\kappa$ in~\eqref{eq:choice_of_representation2} is at least $\dim_\omega\delta/2$,  in~\Cref{le:EHR-25} we give a good randomised strategy to solve $\Sigma$. This strategy assigns the label $e\in E$ 
to $v\in V$ and the label $d\in D$ to $u\in U$ with probabilities
\begin{equation*}
\Pr(v\mapsto e) \ = 
\sum_{\tau \in \widehat{\gr_1^E}, \tau^e\neq 1}
\sum_{s,t\in N_\tau} \dim_\tau\frac{\left\vert \widehat{\A_{x,y}}(\tau_{s,t})\right\vert^2}{|\tau|}
\end{equation*}
and
\begin{equation*}
\Pr(u \mapsto d)\ = 
\sum_{\rho \in \widehat{\gr_1^D}, \rho^d\neq 1}
\sum_{i,j\in N_\rho} \dim_\rho \frac{\left\vert \widehat{\B_{y,z}}(\rho_{i,j})\right\vert^2}{|\rho|},
\end{equation*}
where $x,y,z\in N_\omega$ are suitable indices found in~\Cref{le:EHR-25}. 
These probabilities are supposed to capture the influence of the $e$-th and
$d$-th coordinates on $\A_{x,y}= \omega_{x,y}\circ (A_v)_\varphi$ and
$\B_{y,z}=\omega_{y,z}\circ \ex{s}B_u( \ \cdot^s)^s$ respectively.
\footnote{This is similar to the notion of influence
in~\cite{Bhangale21:stoc,Austrin09:cc}.}
(At this point it may be helpful to recall that $B_u$ is a function from
$\gr_1^D$ to $\gr_2$ and $s$ is a sign sampled uniformly from $\{-1,1\}$. Thus,
$B_u(\ \cdot^s)^s$ takes an element $b\in\gr_1^D$ and returns $(B_u(b^s))^s$.)
This turns out to be a good randomised assignment for $\Sigma$. That is, 
\begin{equation}
\label{eq:strategy_sketch}
\ex{uv}\left[
\sum_{d\in D}
\Pr(v\mapsto \pi_{uv}(d)) \Pr(u\mapsto d)
\right] \geq \alpha,
\end{equation}
where the expectation is taken uniformly over the edges $\{u,v\}$ of $\Sigma$, and $\alpha$ is the soundness constant appearing in~\Cref{th:soundness}. We are being informal with the usage of the word ``probability'' here: the quantities $\Pr(v\mapsto e)$ and
$\Pr(u \mapsto d)$ may add up to less than $1$, but this is easily fixed by normalising, or by letting our strategy default to the uniform assignment with some positive probability. \par
Let us give some more detail. More precisely, \Cref{le:EHR-25} shows that truncating our assignment probabilities to terms of degree less than $\kappa$ is enough to satisfy this last inequality. Let $\ell\geq 0$. The probabilities $\Pr^{<\ell}(v\mapsto e)$, $\Pr^{<\ell}(u\mapsto d)$ are defined the same way as $\Pr(v\mapsto e)$ and $\Pr(u\mapsto d)$
but considering only representations $\tau, \rho$
of degree less than $\ell$. These modified probabilities can be understood as the ``low-degree influences'' of each coordinate in $\A_{x,y}$ and $\B_{y,z}$. With this notation, in~\Cref{le:EHR-25} we prove that~\eqref{eq:strategy_sketch} holds after replacing 
each assignment probability $\Pr$ with its truncated variant $\Pr^{<\kappa}$. In other words, we prove that

\begin{equation*}
\ex{uv}\left[
\sum_{d\in D}
\sum_{\substack{\rho \in \widehat{\gr_1^D}, \rho^d \neq 1\\
|\rho| < \kappa, \, i,j \in N_\rho
}}
\sum_{\substack{\tau \in \widehat{\gr_1^E}, \tau^{\pi_{uv}(d)} \neq 1\\
|\tau| < \kappa, \, s,t\in N_\tau
}}
\frac{\dim_\tau \left\vert \widehat{\A_{x,y}}(\tau_{s,t})\right\vert^2}{|\tau|}
\frac{\dim_\rho\left\vert \widehat{\B_{y,z}}(\rho_{i,j})\right\vert^2}{|\rho|} \right] \geq \alpha.
\end{equation*} 
This shows that our proposed strategy produces a good randomised assignment for $\Sigma$ and completes the soundness proof.

\section{Proof of~\Cref{th:main}} \label{sec:proof-main}

\subsection{Completeness} \label{sec:completeness}

\begin{proof}[Proof of~\Cref{th:completeness}]
Let $f_U:U \to D$, $f_V: V\to E$ be assignments satisfying the original GLC instance $\Sigma$ - that is to say, these assignments are such that $\pi_{uv}(f_U(u)) = f_V(v)$ for each edge $\pi_{uv}$. From these, we construct the assignments $A_v: \gr_1^E \to \gr_1$ and $B_u: \gr_1^D \to \gr_1$ by setting $A_v$ to be the $f_V(v)$-th projection for each $v\in V$, and $B_u$ to be the $f_U(u)$-th projection for each $u\in U$. Observe that the projections are folded over $\mathrm{id}_{\sgr_1}$, so $A_v= (A_v)_{\mathrm{id}_{\sgr_1}}$. Moreover, note that, for a projection function $B:\gr^D\to \gr$, $\bb \in \gr^D$, and $s \in \{s_1,s_2\}$, we have $B(\bb^s)^s=B(\bb)$.
Then, for each edge $\{u,v\}$, each $\tuple{a}\in \gr_1^E$, $\tuple{b}, \bm{\nu}\in \gr^D_1$
and $\tuple{c}=\tuple{b}^{-1} (\tuple{a} \circ \pi_{uv})^{-1}\bm{\nu}$ we have
\begin{align*}
    (A_v(\ba))_{\mathrm{id}_{\sgr_1}} B_u(\bb^{s_1})^{s_1} B_u(\bc^{s_2})^{s_2} & = \ba(f_V(v) ) \bb (f_U(u)) \bb^{-1}(f_U(u)) (\ba \circ \pi_{uv})^{-1}(f_U(u)) \bm{\nu}(f_U(u))\\
    & = \ba(\pi_{uv}(f_U(u)) )  (\ba(\pi_{uv}(f_U(u)) ))^{-1} \bm{\nu}(f_U(u)) = \bm{\nu}(f_U(u)).
\end{align*}

When all parameters are chosen according to the probabilities in~\Cref{fig:probs-reduction}, this expression equals $\groupid_{\gr_1}$ with probability at least $1-\epsilon$. This is because the probability that a given coordinate of $\bm{\nu}$ equals $1_{\gr_1}$ is at least $1-\epsilon$. Therefore, using~\Cref{eq:payoff_gr1} we obtain $\Phi^{\gr_1}_\Sigma(A,B)\geq 1-\epsilon$. 
\end{proof}

\subsection{Soundness}
\label{sec:soundness}

This section is dedicated to the proof of~\Cref{th:soundness}. By assumption, there are families $A, B$
satisfying that
$\Phi_\Sigma^{\gr_2}(A,B)  \geq \frac{1}{|\sgr_2|} + \delta$. In other words, using~\Cref{eq:payoff_gr2},
\[
\ex{\substack{uv, \tuple{a},\tuple{b}, \\ \bm{\nu},s_1,s_2}} \Big[ 
\llbracket (A_v)_{\varphi}(\tuple{a}) B_u(\tuple{b}^{s_1})^{s_1} B_u(\tuple{c}^{s_2})^{s_2}    = \groupid_{\gr_2} \rrbracket 
\Big]  \geq \frac{1}{|\sgr_2|} + \delta,
\]
where $\tuple{c}$ stands for
$\tuple{b}^{-1} (\tuple{a} \circ \pi_{uv})^{-1}\bm{\nu}$ and the expectation is taken over the probabilities defined in~\Cref{fig:probs-reduction}.
Our goal is then to construct assignments of the original GLC instance $\Sigma$ that satisfy at least an $\alpha=\delta^2/(4\kappa|\gr_1|^\kappa|\gr_2|^{4})$-fraction of its constraints, where $\kappa=\lceil(\log_2\delta -2)/(\log_2(1-\epsilon))\rceil$.
Define $z:= (A_v)_{\varphi}(\ba)B_u(\bb^{s_1})^{s_1}B_u(\bc^{s_2})^{s_2}$.  Combining~\Cref{le:sum-dim-char} and~\Cref{le:subrepresentations_restriction}, we obtain 
\[ \sum_{\gamma \in \widehat{\gr_2}} \left(\dim_\gamma \chi_{\gamma}(z) - \dim_\gamma \langle \chi_\gamma, \chi_{R_{\gr_2,\sgr_2}} \rangle \right) =\begin{cases}
    |\gr_2| - \frac{|\gr_2|}{|\sgr_2|}  \quad \text{if } z=\groupid_{\gr_2}\\
    - \frac{|\gr_2|}{|\sgr_2|} \quad \text{otherwise.}
\end{cases}\]
Then, we have 
\begin{align*}
 \ex{\substack{uv, \tuple{a},\tuple{b}\\ \bm{\nu},s_1,s_2}}   &
\Big[ 
\sum_{\gamma \in \widehat{\gr_2}} \dim_\gamma \chi_{\gamma}(z) - 
\dim_\gamma \langle \chi_\gamma, \chi_{R_{\gr_2,\sgr_2}} \rangle \Big] \\ = &
\left(|\gr_2| - \frac{|\gr_2|}{|\sgr_2|}\right)
\Pr\left( z = 1_{\gr_2}\right)
- \frac{|\gr_2|}{|\sgr_2|} 
\Pr\left( z \neq 1_{\gr_2}\right)  \\  \geq &
\left(|\gr_2| - \frac{|\gr_2|}{|\sgr_2|}\right) \left( \frac{1}{|\sgr_2|} +\delta \right)
- \frac{|\gr_2|}{|\sgr_2|} 
\left( 1 -  \frac{1}{|\sgr_2|} -\delta \right) \\  = &
\left(
\frac{|\gr_2|}{|\sgr_2|} -
\frac{|\gr_2|}{|\sgr_2|^2} +
|\gr_2|\delta -
\frac{|\gr_2|}{|\sgr_2|} \delta
\right)
-
\left(
\frac{|\gr_2|}{|\sgr_2|} -
\frac{|\gr_2|}{|\sgr_2|^2} -
\frac{|\gr_2|}{|\sgr_2|} \delta
\right) =  
|\gr_2|\delta.
\end{align*}
Hence
\[
\ex{\substack{uv,\tuple{a},\tuple{b}\\ \bm{\nu},s_1,s_2}}  
\Big[ 
\sum_{\gamma \in \widehat{\gr_2}} \dim_\gamma \chi_{\gamma}(z) 
 \Big] \geq
|\gr_2|\delta +
\sum_{\gamma \in \widehat{\gr_2}}
\dim_\gamma \langle \chi_\gamma, \chi_{R_{\gr_2,\sgr_2}}\rangle.
\]
Using that $|\gr_2| = \sum_{\gamma\in \widehat{\gr_2}} \dim_\gamma^2$, by an averaging argument there is some $\omega\in \widehat{\gr_2}$ satisfying
\begin{equation}\label{eq:chosen-rep}
\left| \ex{\substack{uv,\tuple{a},\tuple{b}\\ \bm{\nu},s_1,s_2}}\left[ \chi_{\omega}\left((A_v)_{\varphi}(\ba)B_u(\bb^{s_1})^{s_1}B_u(\bc^{s_2})^{s_2}) \right) \right] \right| \geq  \dim_\omega\delta 
+ \langle \chi_\omega, \chi_{R_{\gr_2,\sgr_2}} \rangle.
\end{equation}
Observe that such $\omega$ cannot be trivial. Indeed, if $\omega$ is trivial, then the left-hand-side of the above inequality equals $1$, whereas the right-hand-side is strictly greater than $1$. This is because
$\langle \chi_\omega, \chi_{R_{\gr_2,\sgr_2}} \rangle$ equals the multiplicity of the trivial representation in $\omega\vert_{\sgr_2}$, which is one when $\omega$ is trivial itself.

For the rest of this section, we fix such $\omega\in \widehat{\gr_2}$.
We define the short-hands 
\begin{equation} \label{eq:shorthands}
    \A(\ba):= \omega \circ (A_v)_\varphi(\ba), \quad \quad \B(\bb):= \ex{s} [ ((\omega \circ B_u)(\bb^s))^s],
\end{equation}
where $\ba \in \gr_1^E$, $\bb \in \gr_1^D$, and $s$ is sampled uniformly at random from $\{-1,1\}$. Observe that $\A$ depends on the hidden parameter $v$, and $\B$ depends on $u$. Using this notation, and additionally leaving the dependence of $\pi$ on $\{u,v\}$ implicit, we can rewrite \eqref{eq:chosen-rep} as 
\begin{equation}
    \label{eq:chosen-rep-v2}
\left| \ex{\substack{uv,\\ \tuple{a}, \bm{\nu}}}\left[ \tr\left( \A(\ba)(\B*\B)((\ba\circ\pi)^{-1} \bm{\nu}) \right) \right] \right| \geq  \dim_\omega\delta 
+ \langle \chi_\omega, \chi_{R_{\gr_2,\sgr_2}} \rangle.    
\end{equation}
This inequality is the starting point for showing our soundness bound. We briefly describe the proof strategy for~\Cref{th:soundness}. We consider the Fourier series of $\A$ and $(\B*\B)$ in the left-hand-side of \eqref{eq:chosen-rep-v2}. We bound the contribution of terms corresponding to the trivial representation, and terms corresponding to representations 
of high degree. 
This is achieved in~\Cref{le:soundness-aux-1} and~\Cref{le:EHR-23} respectively. Using these results, we show that the contribution of representations 
of low degree is bounded away from zero. Finally, using this fact we are able to construct an assignment to the original GLC instance $\Sigma$ attaining the desired value. This is done in~\Cref{le:EHR-25}.\par

Let us state our main auxiliary results.

\begin{lemma}
\label{le:soundness-aux-1}
Let $\omega$ be as in \eqref{eq:chosen-rep} and $\A$ and $\B$ as in \eqref{eq:shorthands}. Then
\[
\left| \ex{\substack{uv\\ \tuple{a}, \bm{\nu}}}\left[ \tr\left( \widehat{\A}(1)(\B*\B)((\ba\circ\pi)^{-1} \bm{\nu}) \right) \right] \right| \leq  \langle \chi_\omega, \chi_{R_{\gr_2,\sgr_2}} \rangle.  \]
And, in consequence,
\[
\left| \ex{\substack{uv \\ \tuple{a}, \bm{\nu}}}\left[
\tr\left(
\left(
\sum_{\tau\in \widehat{\gr_1^E}, \tau\neq 1
}
\sum_{
s,t\in N_\tau
}
\dim_\tau\widehat{\A}(\tau_{s,t}) \tau_{s,t}(\tuple{a}) \right)(\B*\B)((\ba\circ\pi)^{-1} \bm{\nu}) \right) \right] \right| \geq  \dim_\omega \delta.  \]
\end{lemma}

\begin{lemma} \label{le:EHR-23}
Let $\omega$ be as in \eqref{eq:chosen-rep} and $\A$ and $\B$ as in \eqref{eq:shorthands}. Then

\begin{align}
& \nonumber  \left| \ex{uv,\ba} \left[ 
     \tr \left(
    \left(
    \sum_{\tau\in \widehat{\gr_1^E}, \tau\neq 1}
    \sum_{s,t\in N_\tau}
\dim_\tau
\widehat{\A}(\tau_{s,t}) \tau_{s,t}(\tuple{a}) \right) \times
\right. \right. \right. 
\\ 
& \label{eq:le23_main}
\left. \left. \left.
\left(
\sum_{\rho \in \widehat{\gr_1^D}, |\rho| \geq \kappa}
\sum_{i,j \in N_\rho}
\dim_\rho (1-\epsilon)^{|\rho|} \widehat{(\B*\B)}(\rho_{i,j}) \rho_{i,j}( (\tuple{a} \circ \pi)^{-1})  \right) \right) \right] \right| \leq (\dim_\omega \delta)/2
\end{align}
for any $\kappa\geq (\log_2 \delta - 2)/\log_2(1-\epsilon)$.
\end{lemma}

\begin{lemma}
\label{le:EHR-25}
Let $\omega$ be as in \eqref{eq:chosen-rep} and $\A$ and $\B$ as in \eqref{eq:shorthands}.
Suppose that for some $\kappa>0, \xi>0$ it holds that
    \begin{align*}&
     \left| \ex{uv,\ba}\left[ \tr \left(
    \left( \sum_{\tau\in \widehat{\gr_1^E}, \tau\neq 1}
    \sum_{
s,t\in N_\tau
}
\dim_\tau
\widehat{\A}(\tau_{s,t}) \tau_{s,t}(\tuple{a}) \right) \right.\right.\right. \times &\\ &  
\left.\left.\left.
\left(\sum_{\rho\in \widehat{\gr_1^D}, |\rho| < \kappa}
\sum_{i,j\in N_\rho} \dim_\rho (1-\epsilon)^{|\rho|} \widehat{(\B*\B)}(\rho_{i,j}) \,
\rho_{i,j}((\ba\circ\pi)^{-1})
\right) \right)\right]\right|\geq \xi.   &  
\end{align*}
    Then there is an assignment of the original GLC instance satisfying a
    $\xi^2\kappa^{-1} |\gr_1|^{-\kappa} \dim_{\omega}^{-6}$-fraction of the constraints. 
\end{lemma}

Having stated these auxiliary lemmas, we are prepared to prove our soundness result. 

\begin{proof}[Proof of~\Cref{th:soundness}]
Let $\omega$ be as in \eqref{eq:chosen-rep} and $\A$ and $\B$ as in \eqref{eq:shorthands}. By~\Cref{le:soundness-aux-1}, 
\[
\left| \ex{\substack{uv \\ \tuple{a}, \bm{\nu}}}\left[
\tr\left(
\left(
\sum_{\tau\in \widehat{\gr_1^E}, \tau\neq 1}
\sum_{s,t\in N_\tau}
\widehat{\A}(\tau_{s,t}) \tau_{s,t}(\tuple{a}) \right)(\B*\B)((\ba\circ\pi)^{-1} \bm{\nu}) \right) \right] \right| \geq  \dim_\omega \delta.  \]
Applying the Fourier inversion formula to $W(\tuple{a})= \ex{\bm{\nu}}[(\B*\B)((\ba\circ\pi)^{-1}\bm{\nu})]$ and using~\Cref{le:epsilon-noise}, the above inequality can be rewritten as
\begin{align*}&
     \left| \ex{uv,\ba}\left[ \tr \left(
    \left( \sum_{\tau\in \widehat{\gr_1^E}, \tau\neq 1}
    \sum_{
s,t\in N_\tau
}
\dim_\tau
\widehat{\A}(\tau_{s,t}) \tau_{s,t}(\tuple{a}) \right) \right.\right.\right. \times &\\ &  
\left.\left.\left.
\left(\sum_{\rho\in \widehat{\gr_1^D}}
\sum_{i,j\in N_\rho} \dim_\rho (1-\epsilon)^{|\rho|} \widehat{(\B*\B)}(\rho_{i,j}) \,
\rho_{i,j}((\ba\circ\pi)^{-1})
\right) \right)\right]\right|\geq \dim_\omega \delta.&    \end{align*}
Applying~\Cref{le:EHR-23}, we obtain
    \begin{align*}&
     \left| \ex{uv,\ba}\left[ \tr \left(
    \left( \sum_{\tau\in \widehat{\gr_1^E}, \tau\neq 1}
    \sum_{
s,t\in N_\tau
}
\dim_\tau
\widehat{\A}(\tau_{s,t}) \tau_{s,t}(\tuple{a}) \right) \right.\right.\right. \times &\\ &  
\left.\left.\left.
\left(\sum_{\rho\in \widehat{\gr_1^D}, |\rho| < \kappa}
\sum_{i,j\in N_\rho} \dim_\rho (1-\epsilon)^{|\rho|} \widehat{(\B*\B)}(\rho_{i,j}) \,
\rho_{i,j}((\ba\circ\pi)^{-1})
\right) \right)\right]\right|\geq \dim_\omega \delta / 2,&    \end{align*}
where $\kappa=\lceil(\log_2 \delta - 2)/\log_2(1-\epsilon) \rceil$. 
This allows us to apply~\Cref{le:EHR-25} with $\xi = \dim_\omega \delta /2$, obtaining an assignment of the original GLC instance with value at least
$\delta^2/(4\kappa|\gr_1|^\kappa \dim_\omega^4)$. Using that $\dim_\omega\leq |\gr_2|$ we obtain the desired bound.
\end{proof}

The rest of the section is dedicated to proving~\Cref{le:soundness-aux-1},~\Cref{le:EHR-23}, and~\Cref{le:EHR-25}.

\begin{proof}[Proof of~\Cref{le:soundness-aux-1}]
    The second part of the statement follows from the first, by applying the Fourier inversion formula to $\A(\tuple{a})$ in \eqref{eq:chosen-rep-v2}. We prove the first part of the statement. \par
    The following chain of identities holds:
    \begin{align*}
    \widehat{\A}(1) & =
    \frac{1}{|\gr_1^E|}
    \sum_{\tuple{g}\in \gr_1^E}
    \A(\tuple{g}) =
    \frac{1}{|\gr_1^E|}
    \sum_{\tuple{g}\in \gr_1^E} \omega \circ (A_v)_\varphi(\tuple{g})
    =
        \frac{1}{|\gr_1^E|}
        \frac{1}{|\sgr_1|}
    \sum_{\tuple{g}\in \gr_1^E}
    \sum_{h\in \sgr_1}
    \omega \circ (A_v)_\varphi(h\tuple{g})
    \\ &=
        \frac{1}{|\gr_1^E|}
        \frac{1}{|\sgr_1|}
    \sum_{\tuple{g}\in \gr_1^E}
    \sum_{h\in \sgr_1}
    \omega \circ \varphi(h) \cdot
    \omega \circ (A_v)_\varphi(\tuple{g})
    \\ & =
    \frac{1}{|\gr_1^E|}
    \frac{1}{|\sgr_1|}
    \sum_{\tuple{g}\in \gr_1^E}
    \sum_{h\in \sgr_2}
    \sum_{h^\prime \in \sgr_1, \varphi(h^\prime)=h}
    \omega(h) \cdot
    \omega \circ (A_v)_\varphi(\tuple{g}) \\ & =
    \frac{1}{|\gr_1^E|}
    \frac{1}{|\sgr_1|}
    \sum_{\tuple{g}\in \gr_1^E}
    \sum_{h\in \sgr_2}
    |\mathrm{Ker}(\varphi)|
    \omega(h) \cdot
    \omega \circ (A_v)_\varphi(\tuple{g})
    \\ & =
    \frac{1}{|\gr_1^E|}
    \frac{1}{|\sgr_2|}
    \sum_{\tuple{g}\in \gr_1^E}
    \sum_{h\in \sgr_2}
    \omega(h) \cdot
    \omega \circ (A_v)_\varphi(\tuple{g})
     \\ & =
     \left(
         \frac{1}{|\sgr_2|}
      \sum_{h\in \sgr_2}
    \omega(h)
    \right)
    \frac{1}{|\gr_1^E|}
    \sum_{\tuple{g}\in \gr_1^E}
\A(\tuple{g}).
\end{align*}
Here the fourth equality uses the fact that $(A_v)_\varphi$ is folded over $\varphi$ and $\omega$ is a homomorphism. The second-to-last uses the fact that $|\sgr_2| =|\sgr_1|/|\mathrm{Ker}(\varphi)|$. 
    Let us analyse the matrix $F=\frac{1}{|\sgr_2|}
    \sum_{h\in \sgr_2}
    \omega(h)$.
    Firstly, observe that $F$ is Hermitian. Indeed, using the fact that $\omega$ is a unitary representation we obtain
    \[
    F^* = \frac{1}{|\sgr_2|}
    \sum_{h\in \sgr_2}
    \omega^*(h)= \frac{1}{|\sgr_2|}
    \sum_{h\in \sgr_2}
    \omega(h^{-1}) = F.    
    \]
    Because $F$ is Hermitian, there must be a unitary matrix $U$ and a diagonal matrix $D$ satisfying $F=U^* D U$. We show that $F$'s eigenvalues are $1$, with multiplicity  $\langle \chi_\omega, \chi_{R_{\gr_2,\sgr_2}} \rangle$, and $0$, with multiplicity $\dim_\omega - \langle \chi_\omega, \chi_{R_{\gr_2,\sgr_2}} \rangle$. For each $\rho \in \widehat{\sgr_2}$, let $n_\rho$ be the multiplicity of $\rho$ in $\omega\vert_{\sgr_2}$.    By~\Cref{le:complete_reducibility}, there is an invertible matrix $T$ such that
    $\omega(h)=T^{-1}(\bigoplus_{\rho\in \widehat{\sgr_2}} n_\rho \rho(h) )T$ for all $h\in \sgr_2$. This way,
    \[
    F= 
    T^{-1}\left(
    \bigoplus_{\rho\in \widehat{\sgr_2}} n_{\rho} 
    \frac{1}{|\sgr_2|}
    \sum_{h\in \sgr_2} \rho(h)
    \right) T.
    \]
    By~\Cref{cor:group-sum-zero},
    the matrix $ \frac{1}{|\sgr_2|}
    \sum_{h\in \sgr_2} \rho(h)$ is the zero matrix when $\rho\in \widehat{\sgr_2}$ is non-trivial, and equals $1$ (the one-dimensional identity matrix) when $\rho$ is the trivial representation. Hence, $\left(
    \bigoplus_{\rho\in \widehat{\sgr_2}} n_{\rho} 
    \frac{1}{\sgr_2}
    \sum_{h\in \sgr_2} \rho(h)
    \right)$ is a diagonal matrix containing $n_1$ ones across the diagonal, and $d_\omega- n_1$ zeroes. By~\Cref{le:multiplicity_trivial_restriction}, we know that $n_1=\langle \chi_\omega, \chi_{R_{\gr_2,\sgr_2}} \rangle$, showing the claim. \par

Our two claims together show that there is a unitary matrix $U$ such that $F= U^* D U$, where $D$ is the diagonal matrix whose first $\langle\chi_\omega, \chi_{R_{\gr_2,\sgr_2}} \rangle$ diagonal entries are ones and the rest are zeroes. Now we are prepared to prove the lemma. We have that
\begin{align*}
\left| \ex{\substack{uv \\ \tuple{a}, \bm{\nu}}}\left[ \tr\left( \widehat{\A}(1)(\B*\B)((\ba\circ\pi)^{-1} \bm{\nu}) \right) \right] \right| & =
\left|
\ex{\substack{uv, \tuple{a} \\ \tuple{b}, \tuple{g}, \bm{\nu}}} \left[ 
\tr\left(
U^* D U \A(\tuple{g}) \B(\tuple{b}) \B(\tuple{b}^{-1} (\tuple{a}\circ \pi)^{-1} \bm{\nu})
\right)
\right] \right| \\ & =
\left|
\ex{\substack{uv, \tuple{a} \\ \tuple{b}, \tuple{g}, \bm{\nu}}} \left[ 
\tr\left(
 D \Big( U \A(\tuple{g}) \B(\tuple{b}) \B(\tuple{b}^{-1} (\tuple{a}\circ \pi)^{-1} \bm{\nu}) U^* \Big)
\right) 
\right] \right|.
\end{align*}
Inside the last trace operator we have the product $D M$, where 
\[M = U \A(\tuple{g}) \B(\tuple{b}) \B(\tuple{b}^{-1} (\tuple{a}\circ \pi)^{-1} \bm{\nu}) U^*. \] 
The matrix $M$ is a product of unitary matrices, so it is itself a unitary matrix. The trace of $DM$ is the sum of the first $\langle \chi_\omega, \chi_{R_{\gr_2,\sgr_2}} \rangle$ diagonal elements of $M$, which have absolute value at most $1$, so $|\tr(DM)|\leq \langle \chi_\omega, \chi_{R_{\gr_2,\sgr_2}} \rangle$. This proves the result.   
 \end{proof}

The following property of $\B$ will be useful in the proof of~\Cref{le:EHR-23}.

\begin{lemma} \label{le:EHR-15}
    Let $\gr_1,\gr_2$ be finite groups and $\gamma \in \widehat{\gr_2}$ be a unitary representation. Then, for every function $F:\gr_1 \to \gr_2$, the function $H$ defined by $H(g) = \ex{s\in\{-1,1\}}( (\gamma\circ F)(g^s))^s$ is skew-symmetric. 
\end{lemma}
\begin{proof} For every $g\in\gr_1$,
    \begin{align*}
     H(g^{-1}) & = \ex{s\in\{-1,1\}}( (\gamma\circ F)(g^{-s}))^s  \\ 
            & = \frac{1}{2} (\gamma\circ F)(g^{-1})  + \frac{1}{2} ( (\gamma\circ F)(g))^{-1}  \\ 
            & = \frac{1}{2} (((\gamma\circ F)(g^{-1}))^{-1})^*  + \frac{1}{2} ((\gamma \circ F)(g))^{*}  \\ 
            & = ( \ex{s\in\{-1,1\}}( (\gamma\circ F)(g^s))^s)^*  = H(g)^*.
    \end{align*}
\end{proof}

\begin{proof}[Proof of~\Cref{le:EHR-23}]
We state two auxiliary facts first. We shall show that for any $\rho\in \widehat{\gr_1^D}$
\begin{equation}
    \label{eq:le23_aux0}
    \tr\left(
\sum_{i\in N_\rho}
\widehat{(\B*\B)}(\rho_{i,i})
\right) \geq 0,
\end{equation}
meaning that the left-hand-side is a real non-negative number. 
We shall also show that for any $\tuple{g},\ba \in \gr_1^{E}$ and $\rho\in \widehat{\gr_1^D}$, it holds that
\begin{equation}
\label{eq:le23_aux1}
\left\vert \tr
\left(
\mathcal{C}(\tuple{g},\tuple{a},\rho)
\right)
\right\vert \leq 
\tr\left(
\sum_{i\in N_\rho}
\widehat{(\B*\B)}(\rho_{i,i})
\right),
\end{equation}
where the matrix $\mathcal{C}(\tuple{g},\tuple{a},\rho)$ is defined as 
\[\A(\tuple{g})
\sum_{i,j\in N_\rho}
\widehat{(\B*\B)}(\rho_{i,j})\rho_{i,j}((\ba\circ\pi)^{-1}).\]
Let us prove \eqref{eq:le23_main} assuming these facts. First, observe that
\[
\sum_{\substack{\tau\in \widehat{\gr_2^E}, \tau\neq 1\\
s,t\in N_\tau
}}
\dim_\tau
\widehat{\A}(\tau_{s,t}) \tau_{s,t}(\tuple{a}) = 
\A(\tuple{a}) - \frac{1}{|\gr_1^E|} \sum_{\tuple{g}\in \gr_1^E} \A(\tuple{g}).
\]
Hence,
\begin{align*}
     & 
    \left| 
    \tr \left(
    \left( \sum_{\substack{\tau\in \widehat{\gr_2^E}, \tau\neq 1\\
s,t\in N_\tau
}}
\dim_\tau
\widehat{\A}(\tau_{s,t}) \tau_{s,t}(\tuple{a}) \right)
\left(
\sum_{\substack{\rho \in \widehat{\gr_1^D}, |\rho| \geq \kappa\\ i,j \in N_\rho}} \dim_\rho (1-\epsilon)^{|\rho|} \widehat{(\B*\B)}(\rho_{i,j}) \rho_{i,j}( (\tuple{a} \circ \pi)^{-1})  \right) \right) \right|  \\
 \leq \ &
\left|
\tr\left(
\sum_{\substack{\rho \in \widehat{\gr_1^D}\\ |\rho| \geq \kappa}}
\dim_\rho (1-\epsilon)^{|\rho|}
\mathcal{C}(\tuple{a},\tuple{a},\rho))
\right)
\right| + 
\frac{1}{|\gr_1^E|}
\sum_{\tuple{g}\in \gr_1^E}
\left|
\tr\left(
\sum_{\substack{\rho \in \widehat{\gr_1^D}\\ |\rho| \geq \kappa}}
\dim_\rho (1-\epsilon)^{|\rho|}
\mathcal{C}(\tuple{g},\tuple{a},\rho))
\right)
\right|  
 \\ \leq \ &
2 \tr \left(
\sum_{\substack{\rho \in \widehat{\gr_1^D}\\ |\rho| \geq \kappa}}
\sum_{i\in N_\rho} \dim_\rho (1-\epsilon)^{|\rho|} \widehat{(\B*\B)}(\rho_{i,i})
\right).
\end{align*}
By \eqref{eq:le23_aux0}, 
$\tr\left( \sum_{i\in N_\rho}
\widehat{(\B*\B)}(\rho_{i,i})
\right)$ is non-negative for any $\rho\in \widehat{\gr^D}$, so the last expression is at most
\[
2(1-\epsilon)^\kappa
\tr\left(
\sum_{\rho\in \widehat{\gr^D}}
\sum_{i \in N_\rho}\dim_\rho \widehat{(\B*\B)}(\rho_{i,i})
\right)= 2(1-\epsilon)^\kappa \tr \left( (\B*\B)(\groupid) \right)\leq 
2(1-\epsilon)^\kappa \dim_\omega.
\]
Using that $\kappa\geq (\log_2 \delta - 2)/\log_2(1-\epsilon)$, this completes the proof of the result assuming  \eqref{eq:le23_aux0} and \eqref{eq:le23_aux1}.
Now let us show both of these inequalities. We start with \eqref{eq:le23_aux0}. The following chain of identities holds.
\[
\tr\left(
\sum_{i\in N_\rho}
\widehat{(\B*\B)}(\rho_{i,i})
\right) = 
 \tr\left( 
\frac{1}{|\gr_1^D|}
\sum_{\tuple{g}\in \gr_1^D}
(\B*\B)(\tuple{g})\overline{\chi_\rho(\tuple{g})}
\right) = 
\tr\left(
\frac{1}{|\gr_1^D|}
\sum_{\tuple{g}\in \gr_1^D}
(\B*\B)(\tuple{g})\otimes \overline{\rho(\tuple{g})}
\right). \]

To prove \eqref{eq:le23_aux0} it is enough to show that the matrix inside the last trace operator is positive semidefinite. To do this, we express it as the square of a Hermitian matrix as follows
\begin{align*}
\frac{1}{|\gr_1^D|}
  \sum_{\tuple{g}\in \gr_1^D}
(\B*\B)(\tuple{g})\otimes \overline{\rho(\tuple{g})}
& =
\frac{1}{|\gr_1^D|^2}
\sum_{\tuple{h}\in \gr_1^D}
  \sum_{\tuple{g}\in \gr_1^D}
  \left(
\B(\tuple{h})\otimes \overline{\rho(\tuple{h})}\right)
\left(
\B(\tuple{h}^{-1}\tuple{g})\otimes \overline{\rho(\tuple{h}^{-1}\tuple{g})}
\right) \\
& =
\left(
\frac{1}{|\gr_1^D|}
\sum_{\tuple{g}\in \gr_1^D}
\B(\tuple{g})\otimes \overline{\rho(\tuple{g})}\right)^2.
\end{align*}
The matrix inside the last parentheses is Hermitian.  Indeed,
\[
\left(\sum_{\tuple{g}\in \gr_1^D} \B(\tuple{g})\otimes \overline{\rho(\tuple{g})}\right)^* =
\sum_{\tuple{g}\in \gr_1^D} \B(\tuple{g})^*\otimes \overline{\rho(\tuple{g})^*}= 
\sum_{\tuple{g}\in \gr_1^D} \B(\tuple{g}^{-1})\otimes \overline{\rho(\tuple{g}^{-1})},
\]
where the last equality follows from the fact that both $\B$ and $\rho$ are skew-symmetric. \par
Now, let us show \eqref{eq:le23_aux1}. We start with the following chain of identities. 
\begin{align*}
&
    \left\vert \tr
\left(
\A(\tuple{g}) 
\left(\sum_{i,j\in N_\rho}
\widehat{(\B*\B)}(\rho_{i,j})\rho_{i,j}((\ba\circ\pi)^{-1})
\right)
\right)
\right\vert \\= &
\left\vert \tr
\left(
\A(\tuple{g})
 \left(
\frac{1}{|\gr_1^D|}
\sum_{\tuple{h}\in \gr_1^D}
(\B*\B)(\tuple{h})\overline{\chi_\rho(\tuple{h} (\ba\circ\pi))} \right)
\right)
\right\vert \\ = 
&
\left\vert
\tr
\left(
(\A(\tuple{g})\otimes I_{N_\rho}) 
\left(
\frac{1}{|\gr_1^D|}\sum_{\tuple{h}\in \gr_1^D}
(\B*\B)(\tuple{h}) \otimes \overline{\rho(\tuple{h})}
\right)
\left(I_{N_\omega}\otimes \overline{\rho(\ba\circ\pi)}
\right)
\right)
\right\vert \\ = 
&
\left\vert 
\tr
\left(
\left(
\frac{1}{|\gr_1^D|}
\sum_{\tuple{h}\in \gr_1^D}
(\B*\B)(\tuple{h}) \otimes \overline{\rho(\tuple{h})} 
\right)
\left(I_{N_\omega}\otimes \overline{\rho(\ba\circ\pi)}\right)
\left( A(\tuple{g})\otimes I_{N_\rho}\right)
\right)
\right\vert.
\end{align*}
Both $I_{N_\omega}\otimes \overline{\rho(\ba\circ\pi)}$ and 
$A(\bg)\otimes I_{N_\rho}$ are unitary matrices, and, as shown in the proof of \eqref{eq:le23_aux0}, the matrix $
\frac{1}{|\gr_1^D|}
\sum_{\tuple{h}\in \gr_1^D}
(\B*\B)(\tuple{h}) \otimes \overline{\rho(\tuple{h})}$ is positive semidefinite. Hence, the last expression is at most 
\[
\tr\left(\frac{1}{|\gr_1^D|}
\sum_{\tuple{h}\in \gr_1^D}
(\B*\B)(\tuple{h}) \otimes \overline{\rho(\tuple{h})} 
\right) = \tr\left(
\frac{1}{|\gr_1^D|}
\sum_{\tuple{h}\in \gr_1^D}
(\B*\B)(\tuple{h}) \overline{\chi_\rho(\tuple{h})} 
\right) = 
 \tr\left(
\sum_{i\in N_\rho}
\widehat{(\B*\B)}(\rho_{i,i}) 
\right),
\]
as we wanted to show. This completes the proof of the lemma. 
\end{proof}

\begin{proof}[Proof of~\Cref{le:EHR-25}]
This proof is a direct adaptation of the proof of~\cite[Lemma~25]{EHR04:tcs}.
We can rewrite the inequality in the statement as
\begin{align*}
\xi \leq \left|
\ex{uv}\left[
\sum_{\substack{\tau\in \widehat{\gr_1^E}, \tau\neq 1, \\ s,t \in N_\tau}}
 \sum_{\substack{\rho\in \widehat{\gr_1^D}, |\rho| < \kappa, \\i,j\in N_\rho}}
 (1-\epsilon)^{|\rho|} 
 \ex{\ba} \left[ \tau_{s,t}(\ba) 
 \rho^\pi_{i,j}(\ba^{-1}) \right]
\tr\left(
\dim_\tau \dim_\rho
\widehat{\A}(\tau_{s,t})
\widehat{(\B*\B)}(\rho_{i,j})
\right)
\right] \right|.
\end{align*}
By~\Cref{le:similar_representations}, the inner expectation in the above term is zero unless $\tau\sim_\pi \rho$. Additionally, by an averaging argument, there must be indices $x,y \in N_\omega$
such that this term can be bounded above by
\begin{align*}
&
\dim_{\omega}^2\left| \ex{uv}\left[
 \sum_{\substack{\rho\in \widehat{\gr_1^D},
 |\rho| < \kappa, \\ i,j\in N_\rho}}
 \sum_{\substack{\tau\in \widehat{\gr_1^E}, \tau\neq 1, \\ \tau \sim_\pi \rho, s,t\in N_\tau}}
 (1-\epsilon)^{|\rho|} 
  \langle \tau_{s,t}, \rho^\pi_{j,i} \rangle
 \dim_\tau \dim_\rho \widehat{\A_{x,y}}(\tau_{s,t})
\widehat{(\B*\B)}_{y,x}(\rho_{i,j}) 
\right]\right|.
\end{align*}
Here we have used the fact that
$ \ex{\ba} \left[ \tau_{s,t}(\ba) 
\rho^\pi_{i,j}(\ba^{-1}) \right] =  \langle \tau_{s,t}, \rho^\pi_{j,i} \rangle_{\gr^E}$ since $\rho^\pi$ is unitary.
Moreover, by~\Cref{lem:fourier_convolution}, we know that
\[
\widehat{(\B*\B)}_{y,x}(\rho_{i,j})=
\sum_{z\in N_\omega} \sum_{k\in N_\rho} \widehat{{\B}_{y,z}}(\rho_{i,k}) \widehat{{\B}_{z,x}}(\rho_{k,j}).\]
Hence, by another averaging argument, there is an index $z\in N_\omega$ such that the above expectation is at most
\[
\dim_{\omega}^3 \left| \ex{uv}\left[   \sum_{\substack{\rho\in \widehat{\gr_1^D},
 |\rho| < \kappa, \\ i,j,k\in N_\rho}}
 \sum_{\substack{\tau\in \widehat{\gr_1^E}, \tau \neq 1
 \\ \tau \sim_\pi \rho, s,t\in N_\tau}}
 (1-\epsilon)^{|\rho|} 
  \langle \tau_{s,t}, \rho^\pi_{j,i} \rangle
 \dim_\tau \dim_\rho \widehat{\A_{x,y}}(\tau_{s,t})
 \widehat{{\B}_{y,z}}(\rho_{i,k}) \widehat{{\B}_{z,x}}(\rho_{k,j})
 \right] \right|.
\]
Putting everything together, squaring both sides, and using Jensen's inequality we obtain 
\[
\frac{\xi^2}{\dim_\omega^6}
\leq 
\ex{uv}\left[
\left\vert
\sum_{\substack{\rho\in \widehat{\gr_1^D},
 |\rho| < \kappa, \\ i,j,k\in N_\rho}}
 \sum_{\substack{\tau\in \widehat{\gr_1^E}, \tau\neq 1 \\ \tau \sim \rho^\pi, s,t\in N_\tau}}
 (1-\epsilon)^{|\rho|} 
 \langle \tau_{s,t}, \rho^\pi_{j,i} \rangle
 \dim_\tau \dim_\rho \widehat{\A_{x,y}}(\tau_{s,t})
 \widehat{{\B}_{y,z}}(\rho_{i,k}) \widehat{{\B}_{z,x}}(\rho_{k,j}) \right\vert^2 \right].
\]
The Cauchy-Schwartz inequality now yields  
\begin{align} 
\nonumber
\frac{\xi^2}{\dim_\omega^6} \leq \ &
\ex{uv}\left[
\sum_{\substack{\rho\in \widehat{\gr_1^D},
 |\rho| < \kappa, \\ i,j,k\in N_\rho}}
 \sum_{\substack{\tau\in \widehat{\gr_1^E},  \tau\neq 1\\  \tau \sim \rho^\pi, s,t \in N_\tau}}
 \dim_\tau \dim_\rho
 \left\vert
 (1-\epsilon)^{|\rho|} 
 \langle \tau_{s,t}, \rho^\pi_{j,i} \rangle
  \widehat{{\B}_{z,x}}(\rho_{k,j}) \right\vert^2 
 \right. \\
 \label{eq:lem25_2}
 & 
 \left.
 \times
\sum_{\substack{\rho\in \widehat{\gr_1^D},
 |\rho| < \kappa, \\ i,j,k\in N_\rho}}
 \sum_{\substack{\tau\in \widehat{\gr_1^E}, \tau\neq 1, \\ \tau \sim \rho^\pi, s,t\in N_\tau}}
\dim_\tau \dim_\rho  \left\vert
 \widehat{\A_{x,y}}(\tau_{s,t})
 \widehat{{\B}_{y,z}}(\rho_{i,k}) \right\vert^2 \right].
\end{align}
Let us bound the first term of this last product. 
By Plancherel's Theorem (cfr.~\Cref{th:orthogonal_representations}),
\[
\sum_{\substack{\tau\in \widehat{\gr_1^E}, s,t\in N_\tau}}
\dim_\tau
\left\vert
 \langle \tau_{s,t}, \rho^\pi_{j,i} \rangle \right\vert^2= 
 \Vert \rho^\pi_{j,i} \Vert^2.
\]
As $\rho$ is a unitary representation, it must hold that
$\sum_{i\in N_\rho} \Vert \rho^\pi_{j,i} \Vert^2=1$ for any $j\in N_\rho$. This way,
\begin{align*}
&\sum_{\substack{\rho\in \widehat{\gr_1^D},
 |\rho| < \kappa, \\ i,j,k\in N_\rho}}
 \sum_{\substack{\tau\in \widehat{\gr_1^E}, \tau\neq 1, \\ \tau \sim_\pi \rho, s,t \in N_\tau}}
 \dim_\tau \dim_\rho
 \left\vert
 (1-\epsilon)^{|\rho|} 
 \langle \tau_{s,t}, \rho^\pi_{j,i} \rangle 
  \widehat{{\B}_{z,x}}(\rho_{k,j}) \right\vert^2 \leq \\
 &
 \sum_{\substack{\rho\in \widehat{\gr_1^D},
 |\rho| < \kappa, \\ j,k\in N_\rho}}
 \left(\sum_{i\in N_\rho}
 \Vert \rho^\pi_{j,i} \Vert^2 \right)
  \dim_\rho
  \left\vert
 \widehat{{\B}_{z,x}}(\rho_{k,j})
 \right\vert^2
 =
 \sum_{\substack{\rho\in \widehat{\gr_1^D},
 |\rho| < \kappa, \\ j,k\in N_\rho}}
 \dim_\rho
 \left\vert
 \widehat{{\B}_{z,x}}(\rho_{k,j})
 \right\vert^2.    
\end{align*}
Using Plancherel's Theorem we obtain that this quantity is at most $\lVert {\B}_{z,x}\rVert^2$. This norm must be at most one because $\B$ ranges over unitary matrices. Substituting this back in \eqref{eq:lem25_2}, we get
\[
\frac{\xi^2}{\dim_\omega^6} 
\leq 
\ex{uv}
\left[
\sum_{\substack{\rho\in \widehat{\gr_1^D},
 |\rho| < \kappa, \\ i,j,k\in N_\rho}}
 \sum_{\substack{\tau\in \widehat{\gr_1^E}, \tau\neq 1, \\ \tau \sim_\pi \rho, s,t \in N_\tau}}
 \dim_\tau \dim_\rho
 \left\vert
  \widehat{\A_{x,y}}(\tau_{s,t})
 \widehat{{\B}_{y,z}}(\rho_{i,k}) \right\vert^2 \right].
\]
Summing over $j$, and using the fact that $|N_\rho| \leq
|\gr_1|^{|\rho|} \leq |\gr_1|^\kappa$, we obtain
\[
\frac{\xi^2}{\dim_\omega^6 |\gr_1|^\kappa} 
\leq
\ex{uv}
\left[
\sum_{\substack{\rho\in \widehat{\gr_1^D},
 |\rho| < \kappa, \\ i,k\in N_\rho}}
 \sum_{\substack{\tau\in \widehat{\gr_1^E}, \tau\neq 1, \\ \tau \sim_\pi \rho, s,t \in N_\tau}}
  \dim_\tau \dim_\rho
 \left\vert
 \widehat{\A_{x,y}}(\tau_{s,t})
 \widehat{{\B}_{y,z}}(\rho_{i,k}) \right\vert^2\right].
\]
 The facts that $\tau\neq 1$,
 $|\rho|<\kappa$ and $\tau \sim_\pi \rho$ together imply 
 that $1 < |\tau| \leq |\rho| < \kappa$.
 Hence,
\begin{align}
\nonumber
   \frac{\xi^2}{\dim_\omega^6 \kappa|\gr_1|^\kappa} 
& \leq 
\ex{uv}
\left[
\sum_{\substack{\rho\in \widehat{\gr_1^D},
 |\rho| < \kappa, \\ i,k\in N_\rho}}
 \sum_{\substack{\tau\in \widehat{\gr_1^E}, \tau\neq 1, \\ \tau \sim_\pi \rho, s,t \in N_\tau}}
 \frac{
  \dim_\tau \dim_\rho  \left\vert \widehat{\A_{x,y}}(\tau_{s,t})\right\vert^2
   \left\vert
 \widehat{{\B}_{y,z}}(\rho_{i,k})\right\vert^2}{|\rho|} 
 \right] \\
\nonumber
 &
= 
\ex{uv}
\left[
\sum_{e\in E}
\sum_{\substack{\rho\in \widehat{\gr_1^D} |\rho| < \kappa, \\ i,k\in N_\rho}}
\sum_{\substack{\tau\in \widehat{\gr_1^E}, \tau^e\neq 1 \\ \tau \sim_\pi \rho, s,t \in N_\tau}}
 \frac{\dim_\tau \left|\widehat{\A_{x,y}}(\tau_{s,t})\right|^2}{|\tau|}
 \frac{
    \dim_\rho \left\vert
 \widehat{{\B}_{y,z}}(\rho_{i,k})\right\vert^2}{|\rho|} 
 \right] \\
 &
\nonumber
\leq 
\ex{uv}
\left[
\sum_{\substack{e\in E \\ d\in \pi^{-1}(e)}}
\sum_{\substack{\rho\in \widehat{\gr_1^D}, \rho^d \neq 1 \\
 |\rho| < \kappa, i,k\in N_\rho}}
 \sum_{\substack{\tau\in \widehat{\gr_1^E}, \tau^{e} \neq 1, \\ |\tau| < \kappa, s,t \in N_\tau}}
 \frac{\dim_\tau \left|\widehat{\A_{x,y}}(\tau_{s,t})\right|^2}{|\tau|}
 \frac{
    \dim_\rho \left\vert
 \widehat{{\B}_{y,z}}(\rho_{i,k})\right\vert^2}{|\rho|} 
 \right]
\end{align}
To see the last inequality observe that
given $e\in E$ such that $\tau^e\neq 1$, the fact that $\tau\sim_\pi \rho$ implies that there is at least one $d\in \pi^{-1}(e)$ satisfying $\rho^d\neq 1$. Rearranging the last expression we obtain
\begin{equation} 
 \label{eq:prob-success} 
    \frac{\xi^2}{\dim_\omega^6 \kappa|\gr_1|^\kappa}  \leq 
\ex{uv}
\left[
\sum_{d\in D}
\sum_{\substack{\rho\in \widehat{\gr_1^D}, \rho^d \neq 1 \\
 |\rho| < \kappa, i,k\in N_\rho}}
 \sum_{\substack{\tau\in \widehat{\gr_1^E}, \tau^{\pi(d)} \neq 1, \\ |\tau| < \kappa, s,t \in N_\tau}}
 \frac{\dim_\tau \left|\widehat{\A_{x,y}}(\tau_{s,t})\right|^2}{|\tau|}
 \frac{
    \dim_\rho \left\vert
 \widehat{{\B}_{y,z}}(\rho_{i,k})\right\vert^2}{|\rho|} 
 \right].
\end{equation}

This inequality suggests the following randomised strategy to construct an assignment for the original instance $\Sigma$ of Gap Label Cover. Fix indices $x,y,z \in N_\omega$ as described above. Given a vertex $u\in U$, the assignment $d_u$ is chosen randomly by picking a representation $\rho \in \widehat{\gr_1^D}$ 
with probability at least
$\sum_{i,k\in N_\rho} \dim_\rho \vert\widehat{{\B}_{y,z}}(\rho_{i,k} )\vert^2$, which is well-defined by Plancherel's theorem and the fact that $\omega$ is unitary, and then picking $d_u$ uniformly at random among those elements $d\in D$ for which $\rho^d$ is non-trivial if such $d$ exists, and give up otherwise. Similarly, given a vertex $v\in V$, the assignment $e_v$ is chosen by picking a 
non-trivial representation $\tau  \in \widehat{\gr_1^E}$ 
with probability at least
$\sum_{s,t\in N_\tau} \dim_\tau \vert\widehat{\A_{x,y}}(\tau_{s,t})\vert^2$, and then picking $e_v\in E$ uniformly at random among those elements $e\in E$ for which $\tau^e$ is non-trivial. This is also well-defined by Plancherel's theorem and the fact that $\omega$ is unitary, and additionally, note that picking $e_v$ with $\tau^{e_v}$ non-trivial is always possible since \eqref{eq:prob-success} guarantees that there always exists some non-trivial $\tau \in \widehat{\gr_1^E}$ with $\vert\widehat{\A_{x,y}}(\tau_{s,t})\vert>0$. The probability that this random assignment satisfies a given constraint $u \xrightarrow{\pi} v$ in the GLC instance is at least
the joint probability of the following events: (1) the representations $\rho, \tau$ satisfy  $|\rho|, |\tau| < \kappa$,
and
(2) the chosen element $e_v\in E$ satisfies $e_v = \pi(d_u)$.
Because of \eqref{eq:prob-success}, this probability is at least
$\xi^2\kappa^{-1}|\gr_1|^{-\kappa}\dim_\omega^{-6}$. This strategy can be
derandomised to obtain an assignment of $\Sigma$ satisfying at least the same fraction of its constraints, proving the result. 
\end{proof}

\section{Algebraic Approach to Max-PCSPs} \label{sec:algebraic-approach}

The (in)approximability of Promise Constraint Satisfaction Problems (PCSPs) has
been recently studied extensively in~\cite{Barto24:lics}, where the authors
considered the class of problems called \emph{valued} PCSPs. 
Valued promise CSPs provide a framework for studying approximation problems of a joint
qualitative and quantitative nature in a systematic way, and include, as special
cases, (non-valued) CSPs as well as their promise variant,
valued CSPs, approximability of Max-CSPs, (gap variants of) Label Cover problems, and the Unique Games Conjecture. 

Following the success of the so-called  algebraic approach for
CSPs~\cite{JeavonsCG97,BulatovJK05} and PCSPs~\cite{BBKO21}, Barto et
al.~\cite{Barto24:lics} studied valued PCSPs from an algebraic viewpoint,
developing a framework for deriving polynomial-time reductions between such problems
by studying the relationship between certain associated algebraic objects called \emph{valued minions of plurimorphisms}. These are, loosely, collections of probability distributions over a set of multi-dimensional symmetries of the corresponding feasibility problem, known as the \emph{minion of polymorphisms}. The basic tenet of the algebraic approach is that the existence of a homomorphism between the valued minions of plurimorphisms of two valued PCSPs gives a polynomial-time reduction between the corresponding computational problems (in the opposite direction). 

\medskip
In this section, we show that the hardness of approximation of promise equations can be viewed as an instance of a homomorphism from the valued minion of plurimorphism of $\eq$ to the valued minion of plurimorphism of the Gap Label Cover problem, thus providing an algebraic viewpoint for the reduction shown in this paper.
In the remainder of this section we introduce the valued PCSP framework and the main algebraic tools to show reductions between valued PCSPs, broadly following the presentation in~\cite{Barto24:lics}. We then show that the proof of~\Cref{th:main} in fact provides the sufficient conditions for the existence of a homomorphism between the appropriate valued minions of plurimorphisms.
To be precise, the version of the framework introduced below is a special case
of valued PCSPs in which constraints are restricted to take values from the set
$\{0, 1\}$. This version models problems such as Max-CSPs and its promise
variant, which will suffice for our purposes.
This allows us to simplify the presentation of some of the tools, particularly
due to the fact that for a $\{0,1\}$-valued relation $\phi$ of arity $N$ on a
set $A$, the feasibility set of $\phi$ is the whole of $A^N$, and hence trivially the set of polymorphisms of the feasibility structures associated with a template $\abcs$ is simply the set of multivariate operations from $A$ to $B$. 

\paragraph{Maximum Promise CSPs}

The class of finite sets is denoted by $\FinSet$. Given a finite set $N$, we denote the set of probability distributions on $N$ by $\Delta N$, and the $N$-ary projection to $n\in N$ by $\proj{N}{n}$. Furthermore, for sets $A$ and $B$, we define $\mathscr{O}(A,B)=\{f \mid A^N\to B\}_{N \in \FinSet}$. 

A (relational) \emph{signature} $\sigma$ is finite collection of \emph{relation symbols}, each of which comes with an associated finite set called its \emph{arity}, denoted $\ar(\phi)$ for each symbol $\phi \in \sigma$.
A $\sigma$\emph{-structure} $\bfa$ consists of a finite set $A$, called the \emph{universe} of $\bfa$, and a relation $\phi^\bfa \subseteq A^{\ar(\phi)}$ for each $\phi \in \sigma$, called the \emph{interpretation} of $\phi$ in $\bfa$.

Let $\sigma$ be a relational signature. A \emph{payoff $\sigma$-formula} over a
finite set of variables $X$ is a formal expression of the form 
\[\Phi = \sum_{i
\in I} w_i \phi_i(\tuple{x}_i)\]
where $I$ is a finite non-empty set, $w_i$ are non-negative rational weights satisfying $\sum_{i\in I}w_i=1$, $\phi_i \in \sig$, and $\tuple{x}_i \in X^{\ar(\phi_i)}$ for all $i \in I$.
    Given additionally a $\sigma$-structure $\bfa$ and a map $h:X\to A$, the \emph{interpretation of $\Phi$ in $\bfa$} is the rational-valued function on $A^X$ given by 
   \[ 
    \Phi^{\bfa}(h) = \sum_{i \in I} w_i \llbracket h\tuple{x}_i \in \phi_i^{\bfa}\rrbracket,
   \] 
    for each $h \in A^X$.
    For $c \in [0,1]$, we say that $\Phi$ is \emph{$c$-satisfiable in $\bfa$} if there exists $h \in A^X$ such that $\Phi^\bfa(h) \geq c$.

Let $\bfa$, $\bfb$ be $\sigma$-structures and $c,s \in [0,1]$ be rational constants called the \emph{completeness} and \emph{soundness} parameters respectively. The \emph{Maximum Promise Constraint Satisfaction Problem} over $(\bfa,\bfb,c,s)$, denoted $\pcsp\abcs$, is the following problem: given an input payoff $\sigma$-formula $\Phi$ over a finite set $X$, accept if $\Phi$ is $c$-satisfiable in $\bfa$, and reject if $\Phi$ is not even $s$-satisfiable in $\bfb$. Clearly, this problem is well-defined whenever $\exists h \, \Phi^{\bfa}(h) \geq c$ implies $\exists h \, \Phi^{\bfb}(h) \geq s$ for every payoff $\sigma$-formula $\Phi$. Quadruples $\abcs$ that satisfy this condition are called \emph{templates}.

When $c=1$, we say that the problem has \emph{perfect completeness}: the only accepted instances are those where there exists some assignment $h$ such that $h\bx_i \in \phi_i^\bfa$ for each $i \in I$.

\begin{remark}\label{re:pcsp-groups} 
The problem $\eq(\gr_1,\gr_2,\varphi,c,s)$ studied in this paper can be phrased as a Maximum Promise CSP. In particular, the relational structures $\mathbf{G}_1^\varphi$, $\mathbf{G}_2^\varphi$ corresponding to
systems of promise equations parametrised by  $(\gr_1, \gr_2, \varphi)$ are constructed as follows.
The universe of $\mathbf{G}_1^\varphi$ is $\gr_1$, and similarly the universe of
$\mathbf{G}_2^\varphi$ is $\gr_2$. For every element $g\in \dom(\varphi)$ and every
triple $(i,j,k) \in \{-1,1\}^3$, the signature of $\mathbf{G}_1^\varphi$ and
$\mathbf{G}_2^\varphi$ contains a ternary relation symbol $\phi_{g,(i,j,k)}$,
which is interpreted as the set $\{ (x,y,z) \in \gr^3_1 \mid x^i y^j z^k = g\}$
in  $\mathbf{G}_1^\varphi$, and as $\{(x,y,z) \in \gr_2^3 \mid x^i y^j z^k  =
\varphi(g) \}$ in  $\mathbf{G}_2^\varphi$.
Then, $\eq(\gr_1,\gr_2,\varphi,c,s)$ is precisely $\pcsp(\mathbf{G}_1^\varphi,\mathbf{G}_2^\varphi,c,s)$. 
It can be easily verified that if $(\gr_1,\gr_2,\varphi)$ is a template (i.e., there is a group homomorphism from $\gr_1$ to $\gr_2$ that extends $\varphi$), and $s \leq c$, then $(\mathbf{G}_1^\varphi,\mathbf{G}_2^\varphi,c,s)$ is a template of the Maximum Promise CSP.

Furthermore, note that $\glc_{D,E}(c,s)$ can also be seen as a Maximum Promise CSP, with the only caveat that the corresponding relational structure is on a 2-sorted domain. We omit details here and refer the reader to~\cite{Barto24:lics} for a thorough treatment of multi-sorted PCSPs.
\end{remark}

\paragraph{Minions}

Let $N,N'$ be finite sets, $\pi:N\to N'$ a function known as a \emph{minor map}, and $f : A^{N}\to B$. The \emph{minor of $f$ given by $\pi$}, denoted $f^{(\pi)}$, is the $N'$-ary function defined by
\[f^{(\pi)}(\bx) = f(\bx \pi)\]
for all $\bx \in A^{N'}$.
A \emph{function minion} $\mnn M$ on a pair of sets $(A,B)$ is a non-empty subset of $\mathscr{O}(A,B)$ that is closed under
taking minors. We denote the set of $N$-ary functions in $\mnn{M}$ by $\mnn{M}^{(N)}$.

 Let $\mnn M $ and $\mnn M'$ be function minions. A \emph{minion homomorphism} from $\mnn M $ to  $\mnn M'$ is a collection of functions $(\xi^{(N)}: \mnn M^{(N)} \to \mnn M'^{(N)})_{N \in \FinSet}$ that preserves taking minors, that is, $\xi^{(N')}(f^{(\pi)})=(\xi^{(N)}(f))^{(\pi)}$ for every $N, N' \in \FinSet$, $f \in \mnn{M}^{(N)}$, and $\pi:N \to N'$.

In the exact setting, each Promise CSP can be associated to a function minion, called the minion of polymorphisms of the corresponding template. Intuitively, this minion contains the higher-dimensional symmetries of the problem that determine its complexity.

\paragraph{Valued minions}

In the approximation setting, minions are replaced with valued minions, which are, broadly speaking, collections of probability distributions over multivariate functions from $A$ to $B$. One peculiarity is that these collections are not indexed by single sets, but rather by all finite collections of finite sets.

More formally, let $\mnn M$ be a function minion and $N$ a finite set. An \emph{$N$-ary weighting of $\mnn M$} is a pair 
\[
\Omega=(\OmegaI,\OmegaO) \quad \textnormal{ where } \quad \OmegaI \in \Delta N, \ \OmegaO \in \Delta \mnn{M}^{(N)}.
\]
A \emph{valued minion over $\mnn M$} is a collection $\vmnn M = (\vmnn M^{(\mathcal{N})})$ indexed by finite families of finite sets $\mathcal{N} = (N_j)_{j \in J}$ such that elements of $\vmnn M^{(\mathcal{N})}$ are families $(\Omega_j)_{j \in J}$ where each $\Omega_j$ is an $N_j$-ary weighting of $\mnn M$.

Let $\vmnn M$, $\vmnn M'$ be valued minions over function minions $\mnn M$ and $\mnn M'$, respectively. A \emph{valued minion homomorphism} $\vmnn M \to \vmnn M'$ is a probability distribution $\Xi$ on the set of minion homomorphisms $\mnn M \to \mnn M'$ such that, for every finite set $J$, every family of finite sets $\mathcal{N} = (N_j)_{j \in J}$, and every $(\Omega_j)_{j \in J} \in \vmnn M^{(\mathcal{N})}$, we have  $(\Xi(\Omega_j))_{j \in J} \in \vmnn M'^{(\mathcal{N})}$, where $\Xi(\Omega_j) = (\OmegaI_j, \Xi(\OmegaO_j))$ and $\Xi(\OmegaO_j)$ is defined naturally as sampling $\xi\sim\Xi$, $f\sim \OmegaO_j$, and computing $\xi(f)$. 

\paragraph{Polymorphisms and plurimorphisms}
The specific valued minion associated to a Max-PCSP 
template contains those
weightings that preserve the approximation factor in the following sense. Let
$\abcs$ be a Max-PCSP template and $\mnn M = \mathscr{O}(A,B)$.  
   \begin{itemize}
   \item Let  $\IO \in \mathbb{Q}_{\geq 0}$. An $N$-ary weighting $\Omega$ of $\mnn M$ is
   a \emph{$\IO$-polymorphism} of $\abcs$ if 
       \[ 
   \ex{f \sim \OmegaO} \phi^{\bfb}( f (M) ) - s \geq \IO ( \ex{n \sim \OmegaI}\phi^{\bfa}(M_n)-c) \ \ \ \forall \phi \in \sigma, \ \forall M \in A^{\ar(\phi) \times N}  
       \] where $M_n$ denotes the $n^{\textnormal{th}}$ column of $M$, and $f$ is applied to $M$ row-wise.
   \item An $N$-ary weighting $\Omega$ of $\mnn M$ is
   a \emph{polymorphism} of $\abcs$ if it is a $\IO$-polymorphism for some $\IO \in \mathbb{Q}_{\geq 0}$. 
    \item A finite family $(\Omega_j)_{j \in J}$ of weightings of $\mnn M$ of arities $\mathcal{N}= (N_j)_{j \in J}$  is an $\mathcal{N}$-ary \emph{plurimorphism}  of $\abcs$ if there exists $\IO \in \mathbb{Q}_{\geq 0}$ such that every $\Omega_j$ is a $\IO$-polymorphism. 
\end{itemize}
The collections of plurimorphisms of $\abcs$ is  denoted by $\plu \abcs$.

In particular, observe that for a Max-PCSP template $\abcs$, the collection $\plu\abcs$ forms a valued minion over $\mathscr{O}(A,B)$.

\paragraph{Reductions via homomorphisms}
The main result from~\cite{Barto24:lics} is the following.

\begin{theorem}[\cite{Barto24:lics}] \label{th:valued-reduction-via-homos}
    Let $\abcs$ and $(\bfa',\bfb',c',s')$ be Max-PCSP templates 
    such that the former one has a ``reject'' instance.
    If there is a valued minion homomorphism from $\plu\abcs$ to $\plu(\bfa',\bfb',c',s')$, then $\pcsp(\bfa',\bfb',c',s') \leq \pcsp\abcs$.     
\end{theorem}

In particular, if we can find a valued minion homomorphism from $\plu\abcs$ to the valued minion of plurimorphisms of Gap Label Cover, then this guarantees - in all but the trivial case where all instances are accepted - that $\pcsp\abcs$ is NP-hard. Finding these valued minion homomorphisms is not generally straightforward. Nonetheless, an adaptation of~\cite[Theorem 5.8]{Barto24:lics} gives us the following sufficient condition for the existence of a valued minion homomorphism to Gap Label Cover with perfect completeness.

\begin{theorem}\label{th:valued-hom}
   Let $\abcs$ be a Max-PCSP template, $\mnn{M} = \mathscr{O}(A,B)$, $D$ and $E$ finite disjoint sets, and $\alpha \in \mathbb{R}$. Suppose that there exist
   \begin{itemize}
      \item mappings $\Lambda_D: \mnn{M}^{(D)} \to \Delta D$ and $\Lambda_E: \mnn{M}^{(E)} \to \Delta E$,
      \item for every $\pi:D\to E$, a payoff formula $\Phi_{\pi}$ over the set of variables $A^{D} \cup A^{E}$, 
  \end{itemize}
  such that for every $\pi:D\to E$
  \begin{enumerate}[label=\textnormal{\arabic*.}]
      \item \label{item:1-valued-hom}$\Phi^{\bfa}_{\pi}(\proj{D}{d},\proj{E}{\pi(d)}) \geq c$ for every $d \in D$, and
       \item \label{item:2-valued-hom} for every $\mathbf{d}\in D^N$, $\mathbf{e}\in E^N$, and $N$-ary polymorphism $\Omega$ of $\abcs$, it holds that \[ 
       \ex{f\sim \Omega^{\textnormal{Out}}}  \left[  \Phi^{\bfb}_{\pi}(f^{(\mathbf{d})},f^{(\mathbf{e})}) \right] \geq s  \ \textnormal{ implies }  \ \ex{f\sim \Omega^{\textnormal{Out}}} \ex{\substack{d \sim \Lambda_D(f^{(\mathbf{d})})\\e\sim \Lambda_E(f^{(\mathbf{e})})}} \left[ \pi(d,e) \right] \geq \alpha.\]
  \end{enumerate}
  Then, there exists a valued minion homomorphism from $\plu\abcs$ to $\plu(\glc_{D,E}(1,\alpha))$.
\end{theorem}

The idea behind this construction is roughly as follows: the stochastic map
$\Lambda_D:\mnn{M}^{(D)}\to \Delta D$ is used to define a sampling procedure
over maps  $\lambda_D:\mnn{M}^{(D)}\to D$, and similarly $\Lambda_E$ is used to  define an independent sampling procedure over maps $\lambda_E:\mnn{M}^{(E)}\to E$. Each pair of such maps $(\lambda_D,\lambda_E)$ can be associated to a minion homomorphism $\xi_{\lambda_D,\lambda_E}:\mnn{M} \to \mathscr{O}(D,D) \times \mathscr{O}(E,E)$ in a natural way (see, e.g., \cite[Lemma 4.4]{BBKO21}), and the valued minion homomorphism from~\Cref{th:valued-hom} samples each homomorphism $\xi_{\lambda_D,\lambda_E}$ according to the joint probability of the maps $\lambda_D$ and $\lambda_E$. Conditions~\ref{item:1-valued-hom} and \ref{item:2-valued-hom} guarantee that this sampling procedure does in fact take plurimorphisms of $\abcs$ to plurimorphisms of $\glc_{D,E}(1,\alpha)$.

The remainder of this section is dedicated to showing that $(\mathbf{G}^\varphi_1,\mathbf{G}^\varphi_2,1-\epsilon,1/|\im(\varphi)|+\delta)$ 
satisfies the conditions of~\Cref{th:valued-hom}, where $\alpha$ is as in~\Cref{th:soundness}. That is:

\begin{corollary}
    There is a valued minion homomorphism from $\plu(\mathbf{G}^\varphi_1,\mathbf{G}^\varphi_2,1-\epsilon,1/|\im(\varphi)|+\delta)$ to $\plu(\glc_{D,E}(1,\alpha))$, where $\alpha=\delta^2/(4\kappa |\gr_1|^4 |\gr_2|^\kappa)$ and $\kappa= \lceil(\log_2 \delta - 2)/\log_2(1-\epsilon)\rceil$.
\end{corollary}

\begin{proof}
The payoff formulas $\Phi_{\pi}$ for each $\pi:D\to E$ are defined as in \eqref{eq:reduction-def}. That is, using the notation defined in~\Cref{re:pcsp-groups},
\begin{equation*}
\Phi_\pi =  \ex{\substack{\tuple{a},\tuple{b}, \bm{\nu}\\ s_1,s_2}} \phi_{g_{\tuple{a}},(1,s_1,s_2)}(\tuple{a}^\dagger,\tuple{b}^{s_1},\tuple{c}^{s_2}),
\end{equation*}
where the expectation for each $\pi$ is taken according to the distribution described in~\Cref{fig:probs-reduction}, and as usual we denote $\tuple{c} =
\tuple{b}^{-1} (\tuple{a} \circ \pi)^{-1}\bm{\nu}$. Then, the proof of~\Cref{th:completeness} tells us exactly that~\Cref{item:1-valued-hom} is satisfied for $c=1-\epsilon$.

To show that~\Cref{item:2-valued-hom} is satisfied, we need to define the mappings $\Lambda_D$, $\Lambda_E$. Let $\pi:D\to E$, $\mathbf{d}\in D^N,\mathbf{e}\in E^N$, and $\Omega$ be an $N$-ary polymorphism of $\abcs$ such that $\ex{f\sim \Omega^{\textnormal{Out}}}  \Phi^{\mathbf{G}^\varphi_2}_{\pi}(f^{(\mathbf{d})},f^{(\mathbf{e})}) \geq s$. By the same argument as in~\Cref{sec:soundness}, there is a representation $\omega \in\widehat{\gr_2}$ satisfying
 \begin{equation*}
\left| \ex{f\sim \Omega^{\textnormal{Out}}} \ex{\substack{\tuple{a},\tuple{b},\bm{\nu}\\ s_1,s_2}}\left[ \chi_{\omega}\left((f^{(\mathbf{e})})_{\varphi}(\ba)f^{(\mathbf{d})}(\bb^{s_1})^{s_1}f^{(\mathbf{d})}(\bc^{s_2})^{s_2}) \right) \right] \right| \geq  \dim_\omega\delta 
+ \langle \chi_\omega, \chi_{R_{\gr_2,\sgr_2}} \rangle.
\end{equation*}
Define \begin{equation} 
    \A(\ba):= \omega \circ (f^{(\be)})_\varphi(\ba), \quad \quad \B(\bb):= \ex{s} [ ((\omega \circ f^{(\bd)})(\bb^s))^s],
\end{equation}
and notice that $\A$ and $\B$ now depend on the hidden parameter $f$, but not on variables $u$ and $v$.

With these definitions, it is easy to see that the proofs of~\Cref{le:soundness-aux-1,le:EHR-23} can be adapted to obtain that 

\begin{align*}
& \nonumber  \left| \ex{f\sim \Omega^{\textnormal{Out}}} \ex{\ba} \left[ 
     \tr \left(
    \left(
    \sum_{\tau\in \widehat{\gr_1^E}, \tau\neq 1}
    \sum_{s,t\in N_\tau}
\dim_\tau
\widehat{\A}(\tau_{s,t}) \tau_{s,t}(\tuple{a}) \right) \times
\right. \right. \right. 
\\ 
& 
\left. \left. \left.
\left(
\sum_{\rho \in \widehat{\gr_1^D}, |\rho| < \kappa}
\sum_{i,j \in N_\rho}
\dim_\rho (1-\epsilon)^{|\rho|} \widehat{(\B*\B)}(\rho_{i,j}) \rho_{i,j}( (\tuple{a} \circ \pi)^{-1})  \right) \right) \right] \right| \geq (\dim_\omega \delta)/2.
\end{align*}

Then, we proceed as in the proof of~\Cref{le:EHR-25} to show that there exist indices $x,y,z \in N_\omega$ which satisfy

\begin{equation} \label{eq:lower-bound}
    \frac{\delta^2}{4 \dim_\omega^4 \kappa |\gr_1|^\kappa} \leq \ex{f\sim \Omega^{\textnormal{Out}}} \left[ \sum_{d\in D}
\sum_{\substack{\rho\in \widehat{\gr_1^D}, \rho^d \neq 1\\
 |\rho| < \kappa, i,k\in N_\rho}}
 \sum_{\substack{\tau\in \widehat{\gr_1^E}, \tau^{\pi(d)}\neq 1, \\ |\tau|<\kappa, s,t \in N_\tau}}
 \frac{
  \dim_\tau  \left\vert \widehat{\A_{x,y}}(\tau_{s,t})\right\vert^2}{|\tau|}
 \frac{ \dim_\rho  \left\vert
 \widehat{{\B}_{y,z}}(\rho_{i,k})\right\vert^2}{|\rho|}
 \right]. 
\end{equation}

Then, for these choices of $\omega \in \widehat{\gr_2}$ and $x,y,z \in N_\omega$, the randomised strategy described in the proof of~\Cref{le:EHR-25} gives a pair of maps $\Lambda_D: \mnn{M}^{(D)} \to \Delta D$ and $\Lambda_E: \mnn{M}^{(E)} \to \Delta E$ such that 
sampling $f$ according to $ \Omega^{\textnormal{Out}}$ and $d \in D$, $e\in E$ according to $\Lambda_D(f^{(\mathbf{d})})$ and $\Lambda_E(f^{(\mathbf{e})})$ respectively gives an expected payoff that is bounded below by \eqref{eq:lower-bound}. That is,
\begin{equation*}
 \frac{\delta}{4 \dim_\omega^4 \kappa |\gr_1|^\kappa} \leq  \ex{f\sim \Omega^{\textnormal{Out}}} \ex{\substack{d \sim \Lambda_D(f^{(\mathbf{d})})\\e\sim \Lambda_E(f^{(\mathbf{e})})}} \pi(d,e).
\end{equation*} Using $\dim_\omega \leq |\gr_2|$ gives the desired result.
\end{proof}

\appendix

\section{Proof of~\Cref{th:main-non-cubic}}
\label{ap:non-cubic}

\THnoncubic*

For the tractability part, if $s/c \leq 1/|\im(\varphi)|$ then the
following algorithm works: Count the number equations that are unsatisfiable
over $\gr_2$. If the fraction of unsatisfiable equations is at most $1-c$,
accept. In this case, the random assignment over $\im(\varphi)$ satisfies at
least a $c/|\im(\varphi)|$-fraction of the equations in $\gr_2$. Otherwise, reject. \par

For the hardness part, let  $0<s\leq c <1$ be such that
$s/c > 1/|\im(\varphi)|$, and let $\epsilon >0$ be a rational number such that $c <
1-\epsilon < s|\im(\varphi)|$. Then, the NP-hardness of $\eq(\gr_1,\gr_2, \varphi, c, s)$ follows from the observation that 
\[
\eq(\gr_1,\gr_2, \varphi, c/(1-\epsilon), s/(1-\epsilon)) \leq_p
\eq(\gr_1,\gr_2, \varphi, c, s).
\]
The reduction is as follows. Given any instance of the first problem, we
construct an instance of the second problem by adding an unsatisfiable equation
with weight $\frac{\epsilon}{1-\epsilon}$ and normalising
all weights. The
result then follows from the NP-hardness of $\eq(\gr_1,\gr_2, \varphi,
c/(1-\epsilon), s/(1-\epsilon))$, which in turn follows from~\Cref{th:main}.

\section{Proofs from~\Cref{sec:background}} \label{ap:fourier-proofs}

In this section we prove the results stated in~\Cref{sec:background}.
For the reader's convenience the results are restated below.

\LEfouriercoefficientaux*

\begin{proof}[Proof of~\Cref{le:fourier_coefficient_aux}]
    The following holds.
    \begin{align*}
     \sum_{i,j\in N_\gamma} \widehat{F}(\gamma_{i,j}) \gamma_{i,j}(g) & =
    \sum_{i,j\in N_\gamma} \frac{1}{|\gr|} \sum_{h\in \gr}
    F(h) \overline{\gamma_{i,j}(h)} 
    \gamma_{i,j}(g) = 
    \frac{1}{|\gr|} \sum_{h\in \gr}
    F(h) 
    \sum_{i,j\in N_\gamma}
    \overline{\gamma_{i,j}(h)} 
    \gamma_{i,j}(g) \\
  & = 
    \frac{1}{|\gr|} \sum_{h\in \gr}
    F(h) 
    \sum_{i,j\in N_\gamma}
    \gamma_{j,i}(h^{-1}) 
    \gamma_{i,j}(g) = \frac{1}{|\gr|} \sum_{h\in \gr} F(h) \chi_\gamma(h^{-1}g).    
    \end{align*}
Here the third equality uses the fact that $\overline{\gamma_{i,j}(h)}=\gamma_{j,i}(h^{-1})$ because $\gamma$ is a homomorphism and $\gamma(h)$ is a unitary matrix. The last equality uses again the fact that $\gamma$ is a homomorphism, so $\gamma(h^{-1}g)= \gamma(h^{-1})\gamma(g)$, and 
\[
\chi_\gamma(h^{-1}g)= \sum_{i\in N_\gamma}
\gamma_{i,i}(h^{-1}g)=
\sum_{i,j\in N_\gamma}
\gamma_{i,j}(h^{-1})\gamma_{j,i}(g).
\]
\end{proof}

\LEmultiplicitytrivialrestriction*

\begin{proof}[Proof of~\Cref{le:multiplicity_trivial_restriction}]

Let $\sgr \leq \gr$ be groups, and let $\beta\in \widehat{\sgr}$. The \emph{representation induced by $\beta$ on $\gr$}, denoted 
$\alpha= \Ind_{\sgr}^\gr \beta$, is defined as follows.
Let $N_\alpha = N_\beta \times (\sgr \backslash \gr)$. Fix a set of representatives $\{ g_1,\dots g_k\}$ of $\sgr \backslash \gr$. 
The matrix entries of $\alpha$ are given by 
\[
\alpha_{(i, \sgr g_r), (j,\sgr g_s)}(g)= \begin{cases}
    0 & \text{if } g_r g g_s^{-1} \not\in \sgr, \text{ and } \\
    \beta(g_r g g_s^{-1}) &\text{otherwise }
\end{cases}
\]
for each $(i, \sgr g_r), (j,\sgr g_s)\in N_\alpha$ and $ g\in \gr$.
We have the following claim.

\begin{claim*}
    Let $\sgr \leq \gr$. Then $R_{\gr,\sgr}$ is equivalent to the representation induced on $\gr$ by the trivial representation of $\sgr$.
\end{claim*}
Let $\alpha= \Ind_\sgr^\gr 1$. Then the index set of $\alpha$ is $\{1\} \times
  \sgr \backslash \gr$, which can be naturally identified with $\sgr \backslash
  \gr$. Under this identification, the matrix entries of $\alpha$ and
  $R_{\gr,\sgr}$ are the same, and the claim holds.
\Cref{le:multiplicity_trivial_restriction} then follows from this claim and the
  following result.\footnote{The formulation in~\cite{Terras_1999} is slightly different. To see that it is equivalent to ours, use~\cite[Proposition 3.2]{Terras_1999} and the fact that we defined the inner product over $\mathcal{L}^2(\gr)$ with an additional normalising factor of $\frac{1}{|\gr|}$.}

\begin{theorem}[Frobenius Reciprocity Law]
Let $\sgr \leq \gr$ be groups, let $\alpha$ be a representation of $\gr$, and $\beta$ a representation of $\sgr$. Then,
\[
\langle 
\alpha, \Ind_{\sgr}^\gr \beta
\rangle_{\gr} =
\langle \alpha\vert_\sgr, \beta \rangle_{\sgr}.
\]
\end{theorem}

Indeed, let $\rho\in \widehat{\gr}$. Then, by the Frobenius Reciprocity Law, if we let $\alpha=\rho$ and $\beta$ be the trivial representation we obtain
\[
\langle 
\rho, R_{\gr,\sgr}
\rangle_{\gr} =
\langle \rho\vert_\sgr, 1 \rangle_{\sgr}.
\]
By~\Cref{le:multiplicity},
the first term of the equality is the multiplicity of $\rho$ in $R_{\gr,\sgr}$,
  and the second is the multiplicity of $1$ in $\rho\vert_\sgr$.
  This completes
  the proof of~\Cref{le:multiplicity_trivial_restriction}.
\end{proof}

\LEsubrepresentationsrestriction*

\begin{proof}[Proof of~\Cref{le:subrepresentations_restriction}]
  By~\Cref{le:dimension}, we know that $R_\gr\simeq \bigoplus_{\rho\in \widehat{\gr}} \dim_\rho \rho$, so
    \[
\sum_{\rho\in \widehat{\gr}}
\dim_\rho \langle \chi_\rho, \chi_{R_{\gr,\sgr}} \rangle = 
\langle \chi_{R_\gr}, \chi_{R_{\gr,\sgr}} \rangle = 
\frac{1}{|\gr|} \chi_{R_\gr}(1_\gr) \overline{ \chi_{R_{\gr,\sgr}}(1_\gr)}= \frac{|\gr|}{|\sgr|}.
\]
Here the second equality uses that $\chi_{R_\gr}(g)=0$ for every $g\neq 1_\gr$.
\end{proof}

\LEfourierconvolution*

\begin{proof}[Proof of~\Cref{lem:fourier_convolution}]
    By definition,
    \begin{align*}
    \widehat{(F * G)}_{x,y}(\gamma_{i,j}) & = 
    \frac{1}{|\gr|^2} 
    \sum_{h,g\in \gr}
(F(h)G(h^{-1}g))_{x,y}\overline{\gamma_{i,j}(g)}  \\ & =
    \frac{1}{|\gr|^2} 
    \sum_{h,g\in \gr}\left(
    \sum_{z\in N}
F_{x,z}(h)G_{z,y}(h^{-1}g)
\right) \left(
\sum_{k\in N_\rho} 
\overline{\gamma_{i,k}(h)} \overline{\gamma_{k,j}(h^{-1}g)} \right) 
\\ & =
    \sum_{z\in N, k\in N_\rho}
  \left(
  \frac{1}{|\gr|}
    \sum_{h\in \gr}
F_{x,z}(h)\overline{\gamma_{i,k}(h)} 
\right) \left(
  \frac{1}{|\gr|}
\sum_{g\in \gr} 
G_{z,y}(g)
\overline{\gamma_{k,j}(g)} \right)\\
& =
    \sum_{z\in N,k\in N_\gamma}
    \widehat{F_{x,z}}(\gamma_{i,k})\widehat{G_{z,y}}(\gamma_{k,j}).
    \end{align*}
\end{proof}

\LEsimilarrepresentations*

\begin{proof}[Proof of~\Cref{le:similar_representations}]
    The first statement is straightforward. Let us prove the second. Let $s=(s_e)_{e\in E},
    t=(t_e)_{e\in E}
    \in N_\tau$, and $i=(i_d)_{d\in D}, j=(j_d)_{d\in D} \in N_\rho$. Then
    \begin{align*}
    \langle \tau_{s,t}, \rho^\pi_{i,j} \rangle_{\gr^E} & =
    \frac{1}{|\gr^E|}
    \sum_{\tuple{g}\in \gr^E} 
    \tau_{s,t}(\tuple{g}) \overline{\rho^\pi_{i,j}(\tuple{g})}
    =
    \frac{1}{|\gr^E|}
    \sum_{\tuple{g} \in \gr^E}
    \prod_{e\in E}
    \tau_{s_e,t_e}^e(\bg(e))
    \left(\prod_{d\in \pi^{-1}(e)} \overline{\rho_{i_d,j_d}^d(\bg(e))}\right) 
    \\
    & =  \prod_{e\in E}
     \frac{1}{|\gr|}
    \sum_{g \in \gr}
\tau_{s_e,t_e}^e(g)    \left(\prod_{d\in \pi^{-1}(e)} \overline{\rho_{i_d,j_d}^d(g)}\right)  =
    \prod_{e\in E}
    \langle 
    \tau_{s_e,t_e}^e,
        \otimes_{d\in \pi^{-1}(e)} \rho_{i_d,j_d}^d
    \rangle_{\gr}. 
        \end{align*}
    As $\tau \not\sim_\pi \rho$, there must be an index $e\in E$ for which $\tau^e$ is non-trivial, but $\rho^d$ is trivial for all $d\in \pi^{-1}(e)$. By~\Cref{th:orthogonal_representations}, $\langle
    \tau_{s_e,t_e}^e, \otimes_{d\in \pi^{-1}(e)} \rho_{i_d,j_d}^d
    \rangle_{\gr}=0$ for this choice of $e$, so 
    $\langle \tau_{s,t}, \rho^\pi_{i,j} \rangle_{\gr^E} =0$, as we wanted to prove. 
\end{proof}

\LEepsilonnoise*

\begin{proof}[Proof of~\Cref{le:epsilon-noise}]
The following chain of identities holds
\begin{align*}
& \frac{1}{|\gr^D|}   \sum_{\ba\in \gr^D} \ex{\bm{\nu}}\left[
F(\ba\cdot\bm{\nu})\right]\overline{\rho(\ba)}
  =  \ex{\bm{\nu}}\left[ \frac{1}{|\gr^D|}  \sum_{\ba\in \gr^D}
F(\ba\cdot\bm{\nu}) \overline{\rho(\ba)} \right] 
=  \ex{\bm{\nu}}\left[ \frac{1}{|\gr^D|}  \sum_{\ba\in \gr^D}
F(\ba) \overline{\rho(\ba \bm{\nu}^{-1})} \right] \\
& = \ex{\bm{\nu}}\left[ \frac{1}{|\gr^D|}  \sum_{\ba\in \gr^D}
F(\ba) \ \overline{\rho(\ba)} 
\ \overline{\rho(\bm{\nu}^{-1})}
\right] 
= \frac{1}{|\gr^D|} \sum_{\ba\in \gr^D} F(\ba)
\overline{\rho(\ba)} \ex{\bm{\nu}}[\overline{\rho(\bm{\nu}^{-1})}] .
\end{align*}

The result now follows from the fact that $\ex{\bm{\nu}}[\tsp{\rho(\bm{\nu})}]=(1-\epsilon)I_{N_{\rho}}$. Let us show this identity. 
  Since $\rho \in \widehat{\gr^D}$ we have that \[\ex{\bm{\nu}} \rho(\bm{\nu}) =
  \ex{\bm{\nu}} \left[ \otimes_{d\in D}\rho^d(\bm{\nu}_d) \right] =
  \otimes_{d\in D} \ex{\bm{\nu}(d)} \rho^d(\bm{\nu}_d).\]
Clearly if $\rho^d$ is trivial, $\rho^d(\bm{\nu}_d)=I_{N_{\rho^d}}$, regardless
  of the value of $\bm{\nu}_d$. In all other cases, by
  applying~\cref{cor:group-sum-zero}, we get that \[\ex{\bm{\nu}_d}
  \rho^d(\bm{\nu}_d)=(1-\epsilon)\rho^d(\groupid_{\gr}) + \epsilon \ex{g \in
  \gr}\rho^d(g) = (1-\epsilon)I_{\dim_{\rho^d}}.\]
  Putting this all together, we get that \[\ex{\bm{\nu}} \rho(\bm{\nu}) =
  (1-\epsilon)^{|\rho|}\otimes_{d\in D} I_{\dim_{\rho^d}} =
  (1-\epsilon)^{|\rho|}I_{\dim_{\rho}}.\]
Similarly, we have that $\ex{\bm{\nu}}[\overline{\rho(\bm{\nu}^{-1})}]=
\ex{\bm{\nu}}[\tsp{\rho(\bm{\nu})}]=(1-\epsilon)I_{\dim_{\rho}}$, completing the proof.
\end{proof}

\bibliographystyle{alphaurl}
\bibliography{biblio}

\end{document}